\documentclass[journal,twoside]{IEEEtran}
\usepackage{authblk}
\usepackage{datetime2}
\usepackage{graphicx}
\usepackage{amsmath}
\usepackage{amsthm}
\usepackage{amssymb}
\usepackage{mathrsfs}
\usepackage{stfloats}
\usepackage{dsfont}
\usepackage{setspace}
\usepackage{epstopdf}
\usepackage{cite}
\usepackage{balance}
\usepackage{xcolor}
\usepackage{booktabs}
\usepackage{pifont}
\usepackage{caption}
\usepackage{subcaption}
\usepackage{bm}
\usepackage{bbm}
\usepackage[linesnumbered, ruled,vlined]{algorithm2e}
\usepackage{algorithmic}
\usepackage{hyperref}  
\usepackage{fancyhdr}

\newtheorem{lemma}{Lemma}

\newcommand{\bs}{\boldsymbol}

\begin{document}

\title{Cooperative UAVs for Remote Data Collection under Limited Communications: An Asynchronous Multiagent Learning Framework}
\fancyhead{IEEE TWC}

\author{Cuong Le, Symeon Chatzinotas, \IEEEmembership{Fellow, IEEE}, and Thang X. Vu, \IEEEmembership{Senior Member, IEEE} 
\thanks{C. Le is with the School of Computing, National University of Singapore. Email: cuonglv@comp.nus.edu.sg}
\thanks{S. Chatzinotas and T. X. Vu are with the Interdisciplinary Centre for Security, Reliability and Trust (SnT) - University of Luxembourg, L-1855 Luxembourg. E-mail: \{symeon.chatzinotas, thang.vu\}@uni.lu. \emph{Corresponding author}: Thang X. Vu.}

\thanks{This work was supported in whole, or in part, by the Luxembourg National Research Fund, ref. C22/IS/17220888/RUTINE and INTER/23/17941203/PASSIONATE. For the purpose of open access, and in fulfilment of the obligations arising from the grant agreement, the author has applied a Creative Commons Attribution 4.0 International (CC BY 4.0) license to any Author Accepted Manuscript version arising from this submission.
Parts of this work has been presented at the 2024 IEEE GLOBECOM Workshops \cite{cuong24}.}
\vspace{-0.5cm}}

\markboth{Accepted to IEEE Transactions on Wireless Communications}{C. Le \MakeLowercase{\textit{et al.}}: Cooperative UAVs for Remote Data Collection: An Asynchronous MARL Framework}



\maketitle

\begin{abstract}
    This paper addresses the joint optimization of trajectories and bandwidth allocation for multiple Unmanned Aerial Vehicles (UAVs) to enhance energy efficiency in the cooperative data collection problem. We focus on an important yet underestimated aspect of the system, where action synchronization across all UAVs is impossible. Since most existing learning-based solutions are not designed to learn in this asynchronous environment, we formulate the trajectory planning problem as a Decentralized Partially Observable Semi-Markov Decision Process and introduce an asynchronous multi-agent learning algorithm to learn UAVs' cooperative policies. Once the UAVs' trajectory policies are learned, the bandwidth allocation can be optimally solved based on local observations at each collection point. Comprehensive empirical results demonstrate the superiority of the proposed method over other learning-based and heuristic baselines in terms of both energy efficiency and mission completion time. Additionally, the learned policies exhibit robustness under varying environmental conditions.
\end{abstract}


\section{Introduction}
The rapid advancement of communication technologies like 5G and the anticipated rollout of 6G networks have paved the way for efficient and scalable deployment of Unmanned Aerial Vehicles (UAVs) in real-life applications. Sectors such as agriculture, environmental monitoring, and disaster management increasingly rely on UAVs to collect data quickly and reliably, thanks to their accessibility and highly probable line-of-sight links to ground terminals. For example, UAVs can be used to regularly collect sensory information of the crop conditions in smart agriculture, or monitoring information of windmills in a windfarm.
However, distributed autonomous decision-making and efficient cooperation under limited onboard energy remain critical challenges that must be addressed to maximize the potential of UAVs and ensure their sustainable use.

As future network generations promise unparalleled capabilities, including AI-driven automation, multiagent reinforcement learning (MARL) has emerged as a key approach to addressing the above challenges. With its ability to adapt to dynamic environments and enable autonomous collaboration, MARL is increasingly applied to a wide range of UAV cooperative planning problems \cite{bayerlein2021multi, oubbati2022synchronizing, li2022learning, chen2022joint, wang2023cooperative}. However, most existing MARL-based solutions assume that all UAVs are well-synchronized, in the sense that they all make decisions at the same time. While this assumption simplifies algorithm design, it is difficult to achieve in real-world environments for several reasons: $i$) different UAVs naturally require varying amounts of time to complete their actions; $ii$) forcing UAVs to wait for one another wastes time and system resources; $iii$) achieving synchronization across all UAVs requires extensive control signaling, either among agents or between agents and an intermediary medium; and $iv$) even with an effective synchronization mechanism, unexpected signaling delays or hardware issues can still disrupt synchronization.  Moreover, timing mismatches can directly undermine the performance of policies learned under synchronous assumptions. For instance, UAVs may be required to meet periodically to exchange information or confirm task completion. If their actions are not perfectly aligned, they may fail to meet as planned, leading to overlapping exploration, inefficient waiting cycles, or, in the worst case, a complete breakdown of the learned policy. Addressing asynchronous decision-making is therefore crucial for extending MARL solutions beyond laboratory simulations and making them viable in physical environments. Despite its significance, to the best of our knowledge, asynchronous decision-making in wireless communications, specifically in multi-UAV planning, has not yet been examined.

In this study, we consider the joint optimization problem of UAVs' trajectories and bandwidth allocation at each data collection point to maximize the overall energy efficiency. Our work stands out from existing studies by focusing on an overlooked aspect of the system, where UAVs have to make decisions independently, without time synchronization with one another. This is motivated by the fact that different collection points have varying amounts of data and different transmission channel qualities. As a result, the time required for data collection at each of these points naturally varies. Additionally, since UAVs operate in remote areas with limited inter-UAV communication, synchronizing actions across all UAVs at every step is impossible. The contributions of this paper are as follows:
\begin{itemize}
    \item We consider the UAV-assisted data collection problem under practical conditions: $i$) the location of available data points and the data size are random and unknown in advance; and $ii$) different collection points have varying amounts of data and transmission channel qualities, resulting in different time requirements at each point. Under these conditions, offline optimization approaches are infeasible, and online solutions must equip the UAVs with the capability to make decisions asynchronously.
    \item We introduce a learning-based solution, where the problem is decomposed into two subproblems: trajectory optimization and bandwidth allocation at each hovering point. While bandwidth allocation can be solved independently by each UAV at each hovering point using any convex optimizer, trajectory optimization is more challenging as it requires UAVs to collaborate without explicit time synchronization. To tackle this challenge, we transform it into a Decentralized Partially Observable Semi-Markov Decision Process (Dec-POSMDP), which allows for varying action durations. We then propose Asynchronous-QMIX, a MARL algorithm specifically designed for \textit{ asynchronous environments}. 
    Our algorithm is based on QMIX algorithm \cite{rashid2020monotonic} - one of the state-of-the-art MARL algorithms for \textit{synchronous environments}.
    We theoretically show that under our asynchronous setting, local deterministic greedy policies remain applicable for UAVs, as long as a monotonic relationship between the global Q-value and local utility values holds. Additionally, we also introduce a state downsampling technique to reduce the state space and enhance the scalability of the proposed algorithm.
    \item Extensive experiments are conducted to validate the effectiveness of the proposed method. The numerical results show that our solution outperforms existing learning-based and heuristic solutions. Moreover, the policies learned by our algorithm also demonstrated robust performance under varying environmental conditions.
\end{itemize}

The rest of this paper is organized as follows. Section~\ref{sec:related_work} provides a brief review of existing studies and their limitations. Section~\ref{sec:system} introduces the system model and problem formulation. Section~\ref{sec:mdp} presents the formulation of the trajectory optimization problem as a Dec-POSMDP. Section~\ref{sec:async_learning} presents the Asynchronous-QMIX algorithm for asynchronous environments. Section~\ref{sec:bandwidth_optimization} optimizes the bandwidth allocation under imperfect channel state information. Performance evaluation and discussions are presented in Section~\ref{sec:simulation}. Finally, conclusions and future work are drawn in Section~\ref{sec:conclusion}.

\emph{Notations}: we use superscripts $n$, $c$, and $i$ to indicate the indices of UAVs, cells, and sensor nodes (SNs), respectively; square brackets  $[k]$ to index the individual decision-making steps of each UAV; and subscript $m$ to index the joint global environment transition steps.

\section{Related work}
\label{sec:related_work}
Research on the deployment of UAVs for data collection has flourished in recent years, with a rich body of literature exploring a breadth of system configurations. Early research started with simple scenarios involving a single UAV, with static and deterministic assumptions, allowing traditional optimization approaches to be applied. Most studies focus on UAV trajectory planning and may jointly consider transmission scheduling \cite{yuan2022joint, zhan2017energy, wang2019energy, you20193d, sun2023max}, SN scheduling \cite{zhan2017energy}, bandwidth allocation \cite{samir2019uav, tran2021uav}, power control \cite{feng2021uav, tran2021uav, du2023time}. The common objectives are  mission completion time minimization \cite{yuan2022joint}, energy minimization at UAV or SNs \cite{zhan2017energy, wang2019energy}, data rate and total collected data maximization \cite{you20193d, feng2021uav, sun2023max}, outage probability~\cite{feng2021uav} and Age-of-Information (AoI) \cite{hu2020aoi} minimization. As missions spanning larger areas reveal the inadequacy of a single UAV, recent works turned the attention to multiple UAVs to meet the scaling requirements \cite{zhan2019completion, zhan2019aerial, zhang2021minimizing, xu2021minimizing,arsham25TOCD}. When the system scale becomes increasingly intricate, traditional optimization methods reveal limitations. Fortunately, reinforcement learning (RL), with its adaptability to complex and dynamic environments, can provide the flexibility needed to cope with practical scenarios. RL approach has been demonstrated as a potential solution for both single-UAV \cite{ding20203d, wang2021trajectory, fan2022ris} as well as multiple-UAV systems \cite{hu2020cooperative, oubbati2021multi, emami2023aoi, messaoudi2023uav}.

Despite the attainment of some promising results, there are still technical gaps separating the above studies from practical systems. In particular, studies often assume deterministic systems where either perfect channel state information (CSI) or data collection demand is known in advance at every possible UAV location \cite{yuan2022joint, zhan2017energy, wang2019energy, you20193d, sun2023max, samir2019uav, tran2021uav,feng2021uav, du2023time,hu2020aoi, zhan2019completion, zhan2019aerial, zhang2021minimizing, xu2021minimizing}. Meanwhile, learning-based studies \cite{ding20203d, wang2021trajectory, fan2022ris, hu2020cooperative, oubbati2021multi, emami2023aoi, messaoudi2023uav} often assume that the central controller and UAVs have the capability to observe the entire collection area in real-time. While these assumptions may facilitate the applications of conventional optimization and learning algorithms, they are not in line with practical scenarios, specifically in the context of remote data collection.

Recent studies have made efforts to tackle the challenges of partial observability, where the problem is often formulated as a Decentralized Partially Observable Markov Decision Process (Dec-POMDP), to which a multi-agent reinforcement learning algorithm is applied. In \cite{bayerlein2021multi}, the authors consider the trajectory planning problem for homogeneous UAVs in dense urban areas to maximize the amount of collected data. The problem is formulated as a Dec-POMDP, and then solved by independent reinforcement learning algorithms in which UAVs have different reward functions and are trained independently, while the coordination between UAVs remains open. The authors of \cite{oubbati2022synchronizing} consider two UAV groups, one for data collection and the other for energy transmission with the objectives of minimizing the AoI and maximizing the transmit energy.
Therein, actor-critic algorithm is used to overcome the challenges of partial observability, where the centralized critic can access the global states of the environment during training. A similar AoI minimization problem is studied in~\cite{li2022learning} where the data generation time of SNs is unpredictable. Under the partial observation assumption, a multi-agent algorithm with value-decomposition network (VDN) is adopted to learn UAVs' policies. In \cite{chen2022joint}, PPO algorithm is employed to optimize UAVs' trajectories, while UAV-user association is handled as a coalition formation game. The authors of \cite{wang2023cooperative} focus on optimizing UAVs' speeds, directions, and SN-selection at each hovering point to minimize average AoI, while considering various practical kinematic constraints. The problem is formulated as a Dec-POMDP, on which the well-known QMIX algorithm is applied to learn UAVs' policies. However, due to the representation of SN-selection action, the action space grows exponentially in the number of SNs, resulting in poor scalability.

Although offering fairly comprehensive solutions, the above-mentioned works still heavily rely on a critical assumption regarding the synchronization of UAVs for the convenience of mathematical modeling and problem-solving. Specifically, the time evolution is discretized into equally small intervals, and all UAVs simultaneously make decisions at the beginning of each interval. This assumption, however, is very optimistic in real-world scenarios, especially with large areas of interest and limited inter-UAV communication range. Moreover, synchronous decision-making is time-inefficient, as simultaneous actions necessitate unnecessary waiting periods due to the varying nature of decision epochs among different UAVs. Furthermore, discretizing the timeline as described can be theoretically problematic, as the decision-making under Dec-POMDP is itself PSPACE-complete, and the complexity grows double-exponentially in the planning horizon~\cite{bernstein2002complexity}.
In addition, the impacts of inter-UAV communications have been overlooked in existing studies. Given its practical feasibility, investigating inter-UAV communications during the mission is essential to enhance the cooperative tactics. 

Regarding asynchronous RL, study~\cite{omidshafiei2015decentralized} investigates a coordination problem for multiple robots in continuous environments. The original problem, formulated as a Dec-POMDP, is undecidable over continuous space without approximations or additional assumptions. To address this, the authors reformulate the problem as a Dec-POSMDP, where the original Dec-POMDP is approximated by a discrete environment with a finite number of asynchronous actions. Under Dec-POSMDP, Monte Carlo Search is used to explore sub-optimal policies. Unfortunately, this method is unsuitable for our work due to its inefficient sampling and inapplicability to non-deterministic environments, particularly when data collection demands are not precisely known beforehand. In~\cite{wang2021reducing}, the authors tackle the problem of bus bunching, where buses arrive at stops too closely, disrupting their schedules. They formulate the bus fleet control problem as an asynchronous reinforcement learning problem and propose an extension of the MADDPG algorithm for asynchronous environments, enabling buses to make decisions independently as they reach stops at different times. Recently, the authors in~\cite{yu2023asynchronous} addressed the challenge of multiple robots collaboratively exploring unknown regions, with the objective of minimizing exploration time. They propose an extension of the multi-agent Proximal Policy Optimization (MAPPO) algorithm tailored for asynchronous settings. 
While these studies aim to address asynchronous decision-making, their algorithms are tailored specifically for the characteristics of their own problems and are not directly applicable to our scenario.
Despite its significance, to the best of our knowledge, asynchronous decision-making has not yet been explored in wireless communication systems.

\begin{figure}
    \centering
    \includegraphics[width=0.8\linewidth]{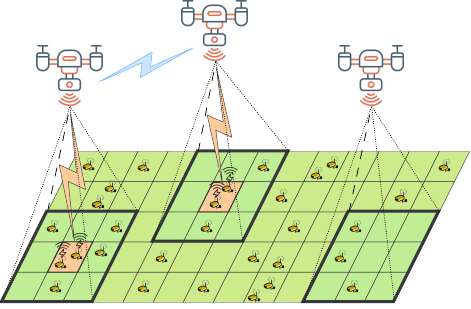}
    \caption{Illustration of the investigated system. Blue links represent inter-UAV communication, and orange links represent SN-UAV data transmission. Cells with dark green and bold borders under UAVs represent their observable regions.}
    \label{fig:system}
    \vspace{-0.4cm}
\end{figure}

\section{System Model and Problem Formulation}
\label{sec:system}
We consider a general data collection problem in which there are $N$ rotary-wing UAVs (not necessarily homogeneous) indexed by $\mathcal{N} = \{1, 2, \ldots, N\}$. Each UAV is equipped with $N_{tx}$ antennas and cooperates with other UAVs to collect data from sensor nodes (SNs) in a given area, as illustrated in Fig.~\ref{fig:system}. The UAV $n$ departs from an initial location $\textbf{w}^n_0 \in \mathbb{R}^2$, travels around the monitored area to explore and collect data from SNs, and then returns to a final destination $\textbf{w}^n_{\text{F}} \in \mathbb{R}^2$. The monitored area is divided into a grid of $H\times H$ equal-sized cells indexed by $\mathcal{C} = \{1, 2, \ldots, H\times H\}$, whose centers are denoted by the set $\mathcal{Q}_{\text{C}} = \{\mathbf{q}_{\text{C}}^c \in \mathbb{R}^2 : c\in \mathcal{C}\}$.
Let $\mathcal{I} = \{1, 2, \ldots, I\}$ be the set of $I$ SNs, where SN $i$ is located at a fixed location $\mathbf{q}_{\text{SN}}^i \in \mathbb{R}^2$. 
%
%

In many UAV-assisted applications, the amount of data to be collected is unknown prior to the UAV deployment, as it depends on dynamic and unpredictable factors. For instance, in search and rescue missions, the required data can fluctuate depending on the discovery of new targets or environmental obstacles encountered during flight. 
Similarly, in wildlife and habitat monitoring, the amount of data collected can be influenced by changing animal behaviors or the need for detailed imagery in specific areas.
To capture this stochastic quantity, let $D^i$ be the amount of data at the $i$th SN that needs to be collected.
The UAVs can move between the centers of cells to explore the data availability, or can hover above these points to collect data.  When a UAV decides to hover and collect data from a cell, it will only leave once all available data in that cell has been collected. The UAV can only communicate with other UAVs within its limited communication range. We also define the termination conditions for a UAV, wherein it concludes the task and flies to its final destination either when it is about to run out of energy or once it has verified that all data have been collected. Furthermore, the termination of each UAV does not affect the operations of the others.

\subsection{UAVs' Trajectories and Energy Consumption}
Let $\mathcal{T}^n = \{{t}^n[0], {t}^n[1], \ldots, {t}^n[K^n]\}$ represent the instants of time at which the $n$-th UAV makes decisions (detailed actions will be defined in Section~\ref{sec:components}) during its flight, and $\mathcal{K}^n = \{0, 1, \ldots, K^n\}$ be the indices of these decisions. Assume that all UAVs operate at the same fixed altitude $h$, the trajectory of the $n$-th UAV can be denoted as $\textbf{w}^n = \{\textbf{w}^n[0], \textbf{w}^n[1], \ldots, \textbf{w}^n[K^n + 1]\}$, where $\textbf{w}^n[k] \in \mathcal{Q}_{\text{C}}$ is the projection on the ground of the UAV's position at time $t^n[k]$. We have the first constraints regarding to the departure and arrival locations of UAVs given by
\begin{equation}
    \label{eq:ctr_init_final_loc}
    \textbf{w}^n[0] = \textbf{w}^n_0,~\textbf{w}^n[K^n + 1] = \textbf{w}^n_{\text{F}}, \forall n \in \mathcal{N}.
\end{equation}
Let $\tau_{\text{H}}^n[k]$ represent the amount of time the $n$-th UAV hovers above $\textbf{w}^n[k]$, and $\tau_{\text{F}}^n[k]$ represent the amount of time required for this UAV to fly from $\textbf{w}^n[k]$ to $\textbf{w}^n[k+1]$.
If, at time $t^n[k]$, the UAV decides to hover above its current location $\textbf{w}^n[k]$ to collect data, we have $\textbf{w}^n[k+1] = \textbf{w}^n[k]$ and $\tau_{\text{F}}^n[k] = 0$. Otherwise, we have $\tau_{\text{H}}^n[k] = 0$ and $\tau_{\text{F}}^n[k] = \frac{\lVert\textbf{w}^n[k+1] - \textbf{w}^n[k]\rVert_2}{v^n}$, where $v^n$ is the fixed velocity of the $n$-th UAV. The total operating time of the $n$-th UAV can then be calculated by
\begin{equation}
    T^n = {\sum}_{k=0}^{K^n}\tau_{\text{F}}^n[k] + {\sum}_{k=1}^{K^n}\tau_{\text{H}}^n[k],
\end{equation}
and the mission completion time is $T = \max_{n\in \mathcal{N}}T^n$.

The $n$-th UAV is powered by an on-board battery with limited capacity of $E^n_{\max}$. Since the communication energy is negligible compared to that required for propulsion \cite{zeng2019energy}, we ignore this component in our analyses. 
Following \cite{zeng2019energy}, the propulsion power consumption of a UAV operating at a fixed altitude is given by
\begin{align}
    \label{eq:uav_energy_model}
    P_{\text{UAV}}(V) = &~P_0\Big(1 + \frac{3V^2}{U_{\text{tip}}^2}\Big)\nonumber+ P_1 \Big(\sqrt{1 + \frac{V^4}{4v_0^4}} - \frac{V^2}{2v_0^2}\Big)^{\frac{1}{2}}\\ & + P_2V^3
\end{align}
where $V$ is the UAV's speed, $U_{\text{tip}}$ is the tip speed of the rotor blade, and $v_0$ represents the mean rotor-induced speed in hover, $P_0, P_1$ and $P_2$ are constants depending on the UAV design and operating environment, which respectively represent the blade profile power, induced power and parasite power. 

Let $E^n[k]$ be the remaining energy of UAV $n$ at the time $t^n[k]$ and all UAVs start with full batteries, i.e., $E^n[0] = E^n_{\max}$. The remaining energy $E^n[k]$ at $t^n[k]$ can be calculated as
\begin{equation}
    E^n[k] = E^n_{\max} - P_{\text{UAV}}(0)\sum_{k'=1}^{k-1}\tau_{\text{H}}^n[k'] - P_{\text{UAV}}(v^n)\sum_{k'=0}^{k-1}\tau_{\text{F}}^n[k'].
\end{equation}
To ensure safety during the mission, we impose the following constraints (for all $n, k$) to guarantee that UAVs always have sufficient energy to reach their final destinations
\begin{align}
    \label{eq:ctr_safe_energy_hover}
    &E^n_{\max} - E^n[k] - P_{\text{UAV}}(0)\tau_{\text{H}}^n[k] \geq \xi^n(\textbf{w}^n[k]) + \epsilon \\
    \label{eq:ctr_safe_energy_move}
    &E^n_{\max} - E^n[k] - P_{\text{UAV}}(v^n)\tau_{\text{F}}^n[k]  \geq \xi^n(\textbf{w}^n[k+1]) + \epsilon
\end{align}
where $\xi^n(\textbf{x}) = \frac{P_{\text{UAV}}(v)\lVert\textbf{x} - \textbf{w}^n_{\text{F}}\rVert_2}{v^n}$ is the energy required to fly to the final destination $\textbf{w}^n_{\text{F}}$ from $\textbf{x}$, and  $\epsilon$ is a small safety energy margin. Let $\textbf{w} = [\textbf{w}^1, \textbf{w}^2, \ldots, \textbf{w}^N]$ be the trajectories of all UAVs and $\textbf{b}$ be the bandwidth allocation strategy (which will be defined more details in the following subsection). The total energy consumed by all UAVs can be calculated by
\begin{equation}
    \label{eq:total_energy}
    \psi(\textbf{w},\! \textbf{b}\!)\! =\! \sum_{n\in\mathcal{N}}\!\sum_{k\in \mathcal{K}^n}\!\left(P_{\text{UAV}}(0)\tau_{\text{H}}^n[k]\! +\! P_{\text{UAV}}(v^n)\tau_{\text{F}}^n[k]\right).
\end{equation}
\\
We note that the bandwidth allocation $\textbf{b}$ directly affects the energy consumption in \eqref{eq:total_energy}.

\subsection{Data Transmission Model}
Let $\bs{h}^{in}(t) \in \mathbb{C}^{N_{tx}\times 1}$ be the channel coefficients between the $n$-th UAV and the $i$-th SN at the time $t$, which can be modeled as
\begin{equation}
    \label{eq:channel_model}
    \bs{h}^{in}(t) = \sqrt{\alpha^{in}(t)}\bs{g}^{in}(t)
\end{equation}
where $\alpha^{in}(t)$ is large-scale fading channel including pathloss and shadowing, and $g^{in}(t)$ is small-scale channels, and $N_{tx}$ is the number of receive antennas of the UAV. 
We adopt the common probabilistic modelling to present the large-scale fading channel, in which the channel is in line-of-sight with a probability of $P_{LoS}^{in}(t)$ and in non-LoS with a probability of $P_{NLoS}^{in}(t) = 1 - P_{LoS}^{in}(t)$. Following  \cite{al2014optimal}, the LoS probability can be approximated as 
\begin{equation}
    P_{LoS}^{in}(t) = \frac{1}{1 + a\exp(-b[\theta^{in}(t) - a])}
\end{equation}
where $a$ and $b$ are environment-dependent parameters, $\theta^{in}(t) = \frac{180}{\pi}\tan^{-1}\left(\frac{H}{\lVert\mathbf{w}^n(t) - \mathbf{q}_{\text{SN}}^i\rVert_2}\right)$ is the elevation angle between the $i$-th SN located at $\mathbf{q}_{\text{SN}}^i$ and the projection $\mathbf{w}^n(t)$ on the ground of the $n$-th UAV at time $t$. 
Therefore, we have
\begin{equation}
    \label{eq:large_fading}
    \alpha^{in}(t) = 
    \begin{cases}
       \alpha_0\left(d^{in}(t)\right)^{-\eta}, &\text{w.p. } P_{LoS}^{in}(t)\\
        \alpha_0\beta\left(d^{in}(t)\right)^{-\eta}, &\text{otherwise}
    \end{cases} 
\end{equation}
where $\alpha_0 = \left({4\pi f_c}/{c}\right)^{-2}$ is the free-space channel power gain at distance of 1m, $d^{in}(t) = \sqrt{\lVert\mathbf{w}^n(t) - \mathbf{q}_{\text{SN}}^i\rVert_2^2 + h^2}$ is the distance between the $i$-th SN and the $n$-th UAV at time $t$, $\eta \geq 2$ is the pathloss exponent, $\beta$ is the attenuation due to NLoS, $f_c$ is the carrier frequency, and $c$ is the speed of light.
We adopt the popular Rician channel model for UAV-SN links \cite{wang2023cooperative,you20193d} which reads:
$\bs{g}^{in}(t) = \bar{g}^{in}\sqrt{\kappa^{in}(t)/(\kappa^{in}(t) + 1)} + \hat{\bs{g}}^{in}\sqrt{1/(\kappa^{in}(t) + 1)}$,
where $\bar{g}^{in}$ is the deterministic LoS component with $|\bar{g}^{in}|~=~1$,  $\hat{\bs{g}}^{in}\sim \mathcal{CN}(0, \bs{1})$ represents NLOS parts, $\kappa^{in}(t)$ is the Rician-factor that depends on the elevation angle $\theta^{in}(t)$ as 
$\kappa^{in}(t) = A_1\exp{\left(A_2\theta^{in}(t)\right)}$, with $A_1$ and $A_2$ being environment-dependent constants. 

When hovering to collect data, the UAVs employ Frequency Division Multiple Access (FDMA) to receive data from all active SNs in the current cell. To mitigate inter-UAV interference, the UAVs  operate in orthogonal frequency bands of the same bandwidth $B = \frac{B_{max}}{N}$, where $B_{max}$ is the system bandwidth. Thus, the bandwidth allocation at one UAV is independent from other UAVs. 
Let $b^{in}(t)$ be the bandwidth that UAV $n$ allocates to active SN $i$ at time $t$ (only active SNs have data to transmit). 
We have the following constraint for $b^{in}(t)$, $\forall n\in \mathcal{N}, t\in \mathcal{T}^n$
\begin{align}
    \label{eq:ctr_B_tot}
    {\sum}_{i\in \mathcal{I}_n(t)}b^{in}(t) \leq B
\end{align}
where $\mathcal{I}_n(t)$ be the set of active SNs serving by UAV $n$ at time $t$. 
The instantaneous achievable upload data rate (bps) from SN $i$ to UAV $n$ at time $t$, assuming maximum ratio combining (MRC) receiver for low complexity and perfect CSI (imperfect CSI will considered in Section~\ref{sec:bandwidth_optimization}), can be calculated as
\begin{equation}
    \label{eq:data_rate}
    R^{in}(t) = b^{in}(t)\log_2\left(1 + {{\|\bs{h}^{in}(t)\|}^2 P_s}/{(b^{in}(t)N_0)}\right)
\end{equation}
where $P_s$ is the fixed transmit power of SNs and $N_0$ is the noise power spectral density.

Consider a particular time $t = t^n[k]$ when the $n$-th UAV makes its $k$-th decision. Due to the strong LoS of the UAV-SN channel when hovering, and for ease of notations, we assume that the channels remain unchanged during UAV's hovering, the hovering time above $\textbf{w}^n[k]$ can then be determined as
\begin{equation}
    \label{eq:hovering_time}
    \tau_{\text{H}}^n[k] =\max_{i\in \mathcal{I}_n(t^n[k])}\left\{\frac{D^i}{R^{in}(t^n[k])} ~\Bigg| ~R^{in}(t^n[k]) > 0\right\}.
\end{equation}
Let $b^{in}[k] = b^{in}(t^n[k])$ and $\textbf{b}^n = \{b^{in}[k] : i\in \mathcal{I}, k\in \mathcal{K}^n\}$ be the bandwidth allocation strategy of UAV $n$.
Given trajectories $\textbf{w}$  and bandwidth allocation strategies $\textbf{b} = [\textbf{b}^1, \textbf{b}^2, \ldots, \textbf{b}^N]$, we have
the total collected data given by
\begin{equation}
    \label{eq:total_data}
    \Phi(\textbf{w}, \textbf{b}) = \sum_{n\in \mathcal{N}}\sum_{k\in \mathcal{K}^n}\sum_{i\in \mathcal{I}_n(t^n[k])} D^i.
\end{equation}
Finally, we have a constraint representing termination conditions of UAVs as follows
\begin{equation}    
    \label{eq:ctr_termination}
    \left(E^n[K^n\! +\! 1]\! -\! \epsilon\right)\!\left({\sum}_{i\in \mathcal{I}}\!D^i\! -\! \Phi(\textbf{w}, \textbf{b})\right)\! \leq\! 0, \forall n\in\mathcal{N}
\end{equation}
where $E^n[K^n + 1]$ is the remaining energy of the $n$-th UAV when arriving at its final destination $\textbf{w}^n_{\text{F}}$. This constraint forces UAVs to continue exploring and collecting data until either their energy falls below the safety level or all active data in the current mission is collected.
\vspace{-0.3cm}
\subsection{Problem Formulation}
Our aim is to jointly optimize the trajectories $\textbf{w}$ and bandwidth allocation strategies $\textbf{b}$ of all $N$ UAVs, such that the overall energy efficiency is maximized. This problem can be mathematically formulated as follows
\begin{align}
    \max_{\textbf{w}, \textbf{b}}  \quad&\frac{\Phi(\textbf{w}, \textbf{b})}{\psi(\textbf{w}, \textbf{b})}\label{prob:P}\tag{P} \\
    \text{s.t.} \quad &\eqref{eq:ctr_init_final_loc}, \eqref{eq:ctr_safe_energy_hover}, 
    \eqref{eq:ctr_safe_energy_move},
    \eqref{eq:ctr_B_tot}, \eqref{eq:ctr_termination}\nonumber
\end{align}
where $\Phi(\textbf{w}, \textbf{b})$ and $\psi(\textbf{w}, \textbf{b})$ are given in \eqref{eq:total_data} and \eqref{eq:total_energy}, respectively. In the above problem, constraint \eqref{eq:ctr_init_final_loc} determines the initial and final locations of the UAVs. Constraints \eqref{eq:ctr_safe_energy_hover} and \eqref{eq:ctr_safe_energy_move} ensure that, at every step, the UAVs have sufficient energy to reach their final locations. Constraint \eqref{eq:ctr_B_tot} guarantees that the total bandwidth used by each UAV does not exceed its available bandwidth. Finally, constraint \eqref{eq:ctr_termination} ensures that a UAV heads to its final location only when all data has been collected or it is about to run out of energy.
Problem \eqref{prob:P} is intractable due to $i)$ high level of uncertainty of CSI and stochasticity of data collection demands; $ii)$ non-deterministicity of the planning horizon $K^n$ and non-convexity in the objective function and constraints; $iii)$ uncertainty in such multi-agent systems, where the operation of each UAV influences the environment and consequently affects the decisions of other UAVs. 
Besides, effective cooperation among UAVs is hindered by limitations in their observation and communication abilities. 
Therefore, it is  difficult, if not impossible, to jointly solve problem \eqref{prob:P} to optimality using conventional methods. 

\section{Trajectory Optimization under Dec-POSMDP}
\label{sec:mdp}
As pointed out in the previous subsection, optimizing the UAVs' trajectory using conventional optimization methods is infeasible due to the lack of information about the SNs' activities as well as the full system states. In this section, we first introduce the Dec-POSMDP framework for trajectory optimization and then define its relevant components in detail.

\subsection{Dec-POSMDP Framework for Trajectory Optimization}
The trajectory optimization problem can be represented by a tuple $\langle\mathcal{N}, \mathcal{S}, \mathcal{A}, P, R, \mathcal{Z}, \gamma\rangle$, where 
$\mathcal{S}$ is the global state space of the environment, $\mathcal{A} = \times_{n\in \mathcal{N}}\mathcal{A}^n$ is the joint action space, $P: \mathcal{S}\times \mathcal{A}\times\mathcal{S}\rightarrow [0, 1]$ is the environment transition kernel, $R:\mathcal{A}\times\mathcal{S}\rightarrow\mathbb{R}$ is a shared reward function contributed by all UAVs, $\mathcal{Z} =  \times_{n\in \mathcal{N}}\mathcal{Z}^n$ is the joint observation space, and finally, $\gamma$ is the discount factor. Note that our proposed solution is model-free, meaning it does not assume or require any knowledge about $P$.
Let $z^n[k] \in \mathcal{Z}^n$ be the local observation received by the $n$-th UAV and $u^n[k] \in \mathcal{A}^n$ be the action chosen by this UAV at time $t^n[k]$. We have the joint action-observation history of the $n$-th UAV until $t^n[k]$ defined as follows
\begin{equation}
    H^n[k] = \left(z^n[1], u^n[1], z^n[2],  \ldots, u^{n}[k-1], z^{n}[k]\right).
\end{equation}
Let $\mathcal{T} = \bigcup_{n\in\mathcal{N}}\mathcal{T}^n = \{\hat{t}_1, \hat{t}_2, \ldots, \hat{t}_m, \ldots\}$ be the set of all decision-making instants of all UAVs sorted in the non-decreasing order. 
Let $s_m$ denote the global state captured at time $\hat{t}_m$ and $\textbf{u}_m = \left(u_m^{1}, u_m^{2}, \ldots, u_m^{N}\right)$ be the joint action of all UAVs taken at this time. The joint action $\textbf{u}_m$ here includes two parts, including a new action calculated by a UAV that finishes its action at $\hat{t}_m$ and the on going actions of other UAVs calculated before $\hat{t}_m$.
Let define the reward function for executing the joint action $\textbf{u}_m$ in state ${s}_{m}$ at time $\hat{t}_m$ by
\begin{equation}
    \label{eq:reward_brief}
    R({s}_{m}, \textbf{u}_{m}) = \int_{\hat{t}_m}^{\hat{t}_{m+1}}\gamma^{t - \hat{t}_m} \widehat{R}(s_m, \textbf{u}_m, t)dt
\end{equation}
where $\widehat{R}(s_m, \textbf{u}_m, t)$ is the total instantaneous reward received by all UAVs at time $t$ for taking action $\textbf{u}_m$ in state $s_m$, which will be detailed in the next subsection. Precisely, ${R}({s}_{m}, \textbf{u}_{m})$ is the total reward accumulated by all UAVs from $\hat{t}_{m}$ to $\hat{t}_{m+1}$. Let $\pi = \times_{n\in \mathcal{N}} \pi^{n}$ be the decentralized joint policy where $\pi^{n}$ is the local policy of UAV $n$ that maps the local action-observation history $H^n[k]$ to the next action $u^n[k]$. Under the policy $\pi$, let define the joint state value function as
\begin{equation}
    {V}_{\texttt{tot}}^{\pi}({s_m}) = \mathbb{E}_{ \pi}\Big[{\sum}_{m=0}^{\infty}\gamma^{\hat{t}_m} {R}({s}_{m}, \textbf{u}_{m}) ~|~ s_0 = s_m\Big],
\end{equation}
and the state-action value function as
\begin{align}
    &Q_{\texttt{tot}}^{\pi}({s_m}, \textbf{u}_m)\nonumber\\
    &\ = \mathbb{E}_{ \pi}\Big[{\sum}_{m=0}^{\infty}\gamma^{\hat{t}_m} {R}({s}_{m}, \textbf{u}_{m}) ~|~ s_0 = s_m, \textbf{u}_0= \textbf{u}_m\Big].
\end{align}
The problem can then be defined as finding the joint policy $\pi^*$ to maximize the values of all states:
\begin{equation}
    \pi^* = \underset{\pi}{\text{argmax}}~ {V}_{\texttt{tot}}^{\pi}({s}_m), \forall {s_m} \in {\mathcal{S}}.
\end{equation}

\subsection{Detailed Components} \label{sec:components}
We now transform the trajectory optimization problem into the Dec-POSMDP by defining its components as follows.

\subsubsection{UAVs' termination conditions} as defined in \eqref{eq:ctr_termination}, the termination of a UAV depends on its energy level and the completion of collecting task. To mitigate the chance of having an empty battery during the flight, before taking actions the UAV checks if the remaining energy approaches the minimum safety level defined in constraints \eqref{eq:ctr_safe_energy_hover} and \eqref{eq:ctr_safe_energy_move}. To handle the second condition, each UAV $n$ maintains a completion map $G^n \in \{0, 1\}^{H^2\times 1}$ indicating its belief about the status of every cell. This map is initialized to zeros at the beginning, indicating the `yet-completed' status. Over time, the map is updated after every UAV action based on its observations. Specifically, $G^{cn}$ is set to~1, or `completed', if the $n$-th UAV collects all data in the $c$-th cell or observes that there is no data in the cell. The task is then considered completed by UAV $n$ if all elements of $G^n$ are equal to one.

\subsubsection{State space} each state $s_m \in \mathcal{S}$ includes positions, energy levels, and completion maps of all UAVs, and the true completion map.

\subsubsection{Observation} the local observation $z^n[k]$ of the $n$-th UAV includes its remaining energy, current position, local completion map $G^n$, and total data collection demand of each cell within its current observable region.

\subsubsection{Action space} 
the action space of the UAV $n$ consists of five actions $\mathcal{A}^n = \{0, 1, 2, 3, 4\}$, which represent the action of hovering to collect data, moving forward, to the right, backward, and to the left, respectively. To enhance sample efficiency and prevent UAVs from taking nonsensical actions, we impose the following constraint: if a UAV moves to a specific location, it is not allowed to immediately return to its previous location. Specifically, the UAV cannot perform two consecutive actions of (1, 3), (3, 1), (2, 4), or (4, 2) in successive time steps. This is achieved by simply masking the prohibited actions during the sampling process.

\subsubsection{Communication message between UAVs} When two UAVs are within their communication range, they share positions and synchronize the completion map to assist each other in verifying the mission completion for efficient exploration.

\subsubsection{Reward function}
as defined in \eqref{eq:reward_brief}, the reward function $R(s_m, \textbf{\textbf{u}}_m)$ under the joint policy $\pi$ is the accumulation of the instantaneous reward $\widehat{R}(s_m, \textbf{u}_m, t)$. Here, we define this instantaneous reward as follows:
\begin{equation}
    \label{eq:imm_reward}
    \widehat{R}(s_m, \textbf{u}_m, t) = \sum_{n \in \mathcal{N}}\sum_{i \in \mathcal{I}}\varphi^{in}(t) + \alpha \Gamma(\pi)\mathbbm{1}(s_m, \textbf{u}_m, t)
\end{equation}
where\footnote{The reward function can be readily applied to the imperfect CSI scenario by using the robust rate formulation derived in \eqref{eq:data_rate_imp_csi} instead of $R^{in}(t)$. Bandwidth allocation is incorporated during the trajectory learning via the reward function.} 
\begin{equation}
    \varphi^{in}(t) = \begin{cases}
        R^{in}(t), &\text{if $u^n_m = 0$}\\
        -0.01, &\text{otherwise.}
    \end{cases}
\end{equation}
In the above equations, $R^{in}(t)$ is the upload data rate given in~\eqref{eq:data_rate}, $\alpha$ is a scaling parameter, $\Gamma(\pi)$ is the energy efficiency achieved under the policy $\pi$ given in the objective function of \eqref{prob:P}, and $\mathbbm{1}(s_m, \textbf{u}_m, t)$ is an indicator function indicating whether taking action $\textbf{u}_m$ in state $s_m$ leads to the termination of the last UAV at time $t$. 
This reward function can be simply interpreted as follows. If taking action $\textbf{u}_m$ in state $s_m$ does not lead to the termination state, the reward $\widehat{R}(s_m, \textbf{u}_m, t)$ at time $t$ is the total reward of all UAVs, where the reward contributed by UAV $n$ is its data collection rate $R^{in}(t)$ if it is collecting data (indicated by its action $u_m^n = 0$), and a small negative reward of $-0.01$ if it is moving between cells. On the other hand, if taking action $\textbf{u}_m$ in state $s_m$ leads to the termination state (i.e., mission completion), the UAVs receive a reward proportional to their energy efficiency. The immediate reward $\varphi^{in}(t)$ in \eqref{eq:imm_reward} serves as a shaping function to support the main learning goal of maximizing energy efficiency, which is only activated at the final moment when the last UAV reaches its destination.  With a sufficiently large scaling parameter $\alpha$ (which is set to $\max_{n\in \mathcal{N}}\{E^n_{\max}\}$ in our experiments), equation \eqref{eq:imm_reward} ensures that UAVs receive a large reward only when all of them complete the mission. Moreover, by assigning a small negative reward to each movement action, UAVs are incentivized to cooperate and complete the mission in as few steps as possible, thereby reducing energy consumption and enhancing energy efficiency.
Combine \eqref{eq:reward_brief} and \eqref{eq:imm_reward}, we obtain the form of the reward function as follows
\begin{align}
    \label{eq:reward_plain}
    &R({s}_{m}, \textbf{u}_{m})\nonumber\\ 
    &= \int_{\hat{t}_m}^{\hat{t}_{m+1}}\gamma^{t - \hat{t}_m}\Big(\sum_{\underset{n \in \mathcal{N}}{i \in \mathcal{I}}}\varphi^{in}(t) + \alpha\Gamma(\pi)\mathbbm{1}(s_m, \textbf{u}_m, t) \Big)dt\nonumber\\
    &= \sum_{\underset{n \in \mathcal{N}}{i \in \mathcal{I}}}\varphi^{in}\left(\hat{t}_m\right)\int_{\hat{t}_m}^{\hat{t}_{m+1}}\gamma^{t - \hat{t}_m}dt + \gamma^{\hat{\tau}_m}\alpha \Gamma(\pi)\hat{\mathbbm{1}}(s_m, \textbf{u}_m)\nonumber \\
    &=\frac{\gamma^{\hat{\tau}_m} - 1}{\ln{\gamma}} \sum_{\underset{n \in \mathcal{N}}{i \in \mathcal{I}}}\varphi^{in}\left(\hat{t}_m\right) + \gamma^{\hat{\tau}_m}\alpha \Gamma(\pi)\hat{\mathbbm{1}}(s_m, \textbf{u}_m)
\end{align}
where $\hat{\tau}_m = \hat{t}_{m+1} - \hat{t}_{m}$ is the duration of action $\textbf{u}_m$ taken in state $s_m$ and $\hat{\mathbbm{1}}(s_m, \textbf{u}_m)$ indicating whether taking action $\textbf{u}_m$ in state $s_m$ leads to the termination of the last UAV. Here, the second equation is due to the assumption that channel states remain unchanged during UAV's hovering, and that $\mathbbm{1}(s_m, \textbf{u}_m, t)$ is only activated at the last moment when the last UAV arrives at its final destination.

\section{Asynchronous Learning Algorithm}
\label{sec:async_learning}
This section introduces an asynchronous learning algorithm for learning UAVs' trajectory policies. To enhance the scalability of the proposed algorithm, we also propose state downsampling method to reduce the state space. Finally, a computational complexity analysis is provided, where we focus on the computational overhead required for each UAV to calculate an action in the deployment.

\subsection{Asynchronous QMIX}
The conventional QMIX algorithm \cite{rashid2020monotonic} aims to learn a centralized action-value function $Q_{\texttt{tot}}(s, \textbf{u})$, which is factorized into $N$ individual utility functions $Q^n(H^n, u^n)$ representing the goodness of taking action $u^n$ on history $H^n$. The principal premise that makes QMIX efficient is the consistent relationship between the deterministic greedy centralized policy and the deterministic greedy decentralized policies, which results from the monotonicity between $Q_{\texttt{tot}}$ and $Q^n$, i.e., $\partial Q_{\texttt{tot}} / \partial Q^n \geq 0, \forall n \in \mathcal{N}$. When such monotonicity is assured, the centralized training can be executed relying on the following relation (neglecting the order of UAVs)
\begin{equation}
    \underset{\textbf{u}\in \mathcal{A}}{\text{argmax}}Q_{\texttt{tot}}(s, \textbf{u}) = \left\{\underset{u^n \in \mathcal{A}^n}{\text{argmax}}Q^n(H^n, u^n)\right\}_{n\in \mathcal{N}}.
\end{equation}
This result implies that local actions that improve the local $Q^n$ values will also enhance the joint action-value function $Q_{\texttt{tot}}$, enabling decentralized agents to operate independently based on the greedy policy applied to their local $Q^n$ values.

It is worth noting that the above-mentioned original QMIX algorithm is only designed for synchronous environments in which all agents take actions simultaneously at each time step. As will be shown in Section~\ref{sec:result - performance comparison}, when this assumption does not hold, its performance significantly degrades. 
To enable QMIX to be applicable in our asynchronous
environment, we make an important modification to the QMIX algorithm as follows.
Let $\hat{n}$ be the UAV that finishes its action at time $\hat{t}_m$, and $H^{n}_m$ be the local joint action-observation history of UAV $n$ recorded until $\hat{t}_m$ (i.e., $H^{n}_m = H^n[k^*]$ where $k^* = \text{argmax}_{k\in \mathcal{K}^n}\{t^n[k]~|~t^n[k] \leq \hat{t}_m\}$).
As mentioned in the preceding section, there is a new action $u^{\hat{n}}_m$ calculated by UAV $\hat{n}$ at $\hat{t}_m$, while the actions of other UAVs being executed. Let $\textbf{u}^{-\hat{n}}_m$ denote the set of these ongoing actions, we have the joint action at time $\hat{t}_m$ given by $\textbf{u}_m = \{u^{\hat{n}}_m\} \cup\textbf{u}^{-\hat{n}}_m$. The aim of the asynchronous algorithm is to learn a joint (centralized) action-value function $Q_{\texttt{tot}}(s_m, \textbf{u}_m | \textbf{u}^{-\hat{n}}_m)$ conditioned on the ongoing actions $\textbf{u}^{-\hat{n}}_m$. We have the following result for our asynchronous algorithm.
\begin{lemma}
    \label{lemma:aqmix_monotonic}
    Given that $\frac{\partial Q_{\texttt{tot}}}{\partial Q^n} \geq 0\ \forall n \in \mathcal{N}$, and that
    \begin{align*}
        Q_{\texttt{tot}}&\left(s_m, \textbf{u}_m | \textbf{u}^{-\hat{n}}_m\right)\\ &= Q_{\texttt{tot}}\left(Q^{\hat{n}}(H^{\hat{n}}_m, u^{\hat{n}}_m), \{Q^n(H^n_m, u_m^n)\}_{n\in \mathcal{N}\setminus\{\hat{n}\}}\right),
    \end{align*}
    then we have
    $$\underset{\textbf{u}_m \in \mathcal{A}}{\normalfont\text{argmax}}Q_{\texttt{tot}}\left(s_m, \textbf{u}_m | \textbf{u}^{-\hat{n}}_m\right) =\{\underset{u^{\hat{n}}_m \in \mathcal{A}^{\hat{n}}}{\normalfont\text{argmax}}Q^{\hat{n}}\left(H^{\hat{n}}_m, u^{\hat{n}}_m\right)\} \cup\textbf{u}^{-\hat{n}}_m.$$
\end{lemma}
\begin{proof}
    Since $\frac{\partial Q_{\texttt{tot}}}{\partial Q^n} \geq 0\ \forall n \in \mathcal{N}$, we have
    \begin{align*}
        Q&_{\texttt{tot}}\left(Q^{\hat{n}}(H^{\hat{n}}_m, u^{\hat{n}}_m), \{Q^n(H^n_m, u_m^n)\}_{n\in \mathcal{N}\setminus\{\hat{n}\}}\right)\\
        &\leq Q_{\texttt{tot}}\left(\underset{u^{\hat{n}}_m \in \mathcal{A}^{\hat{n}}}{\max}Q^{\hat{n}}(H^{\hat{n}}_m, u^{\hat{n}}_m), \{Q^n(H^n_m, u_m^n)\}_{n\in \mathcal{N}\setminus\{\hat{n}\}}\right)\\
        &=\underset{u^{\hat{n}}_m\in \mathcal{A}^{\hat{n}}}{\max}Q_{\texttt{tot}}\left(Q^{\hat{n}}(H^{\hat{n}}_m, u^{\hat{n}}_m), \{Q^n(H^n_m, u_m^n)\}_{n\in \mathcal{N}\setminus\{\hat{n}\}}\right).
    \end{align*}
    Moreover, since an agent only receives a new observation, updates local history and calculates a new action when an action is finished, $Q^n(H^n_m, u^n_m)$ remains unchanged for $n\in \mathcal{N}\setminus\{\hat{n}\}$. 
Combined with the definition of $Q_{\texttt{tot}}(s_m, \textbf{u}_m | \textbf{u}^{-\hat{n}}_m)$, we have
    \begin{align*}
        \underset{\textbf{u}_m \in \mathcal{A}}{\text{max}}&Q_{\texttt{tot}}(s_m, \textbf{u}_m | \textbf{u}^{-\hat{n}}_m)\\ 
        &= \underset{\textbf{u}_m \in \mathcal{A}}{\text{max}}Q_{\texttt{tot}}\left(Q^{\hat{n}}(H^{\hat{n}}_m, u^{\hat{n}}_m), \{Q^n(H^n_m, u_m^n)\}_{n\in \mathcal{N}\setminus\{\hat{n}\}}\right)\\
        &=\underset{u^{\hat{n}}_m \in \mathcal{A}^{\hat{n}}}{\max}Q_{\texttt{tot}}\left(Q^{\hat{n}}(H^{\hat{n}}_m, u^{\hat{n}}_m), \{Q^n(H^n_m, u_m^n)\}_{n\in \mathcal{N}\setminus\{\hat{n}\}}\right)
    \end{align*}
    which directly implies the result and completes the proof.
\end{proof}
\noindent This lemma implies that once the monotonicity between $Q_{\texttt{tot}}$ and $Q^n$ is established, the local deterministic greedy policies remain applicable for agents to calculate their actions in asynchronous environments.

\begin{figure*}
    \centering
    \includegraphics[width=0.8\linewidth]{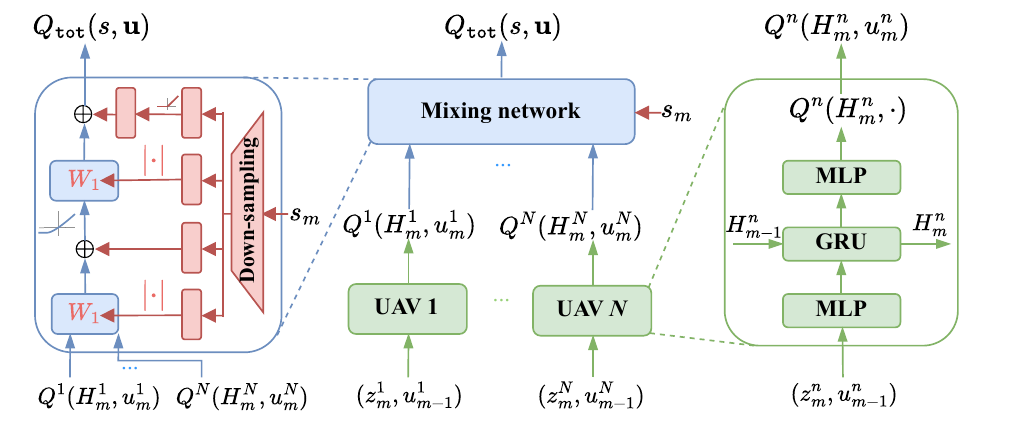}
    \caption{The architecture of the proposed algorithm. The green, blue, and red blocks represent UAVs' policy networks, mixing network, and the hypernetwork. A downsampling layer is included in the hypernetwork to reduce the state space.}
    \label{fig:aqmix}
    \vspace{-0.5cm}
\end{figure*}

Fig.~\ref{fig:aqmix} illustrates the network architecture in our proposed algorithm. Specifically, the architecture includes two components: i) agent networks that take local action-observation histories and the agent indices as input and output the local $Q^n$ values, and ii) a mixing network that uses the global state information to produce $Q_{\texttt{tot}}$ from $N$ local $Q^n$ values.
We employ DRQN~\cite{hausknecht2015deep} with parameter sharing for each agent network, while for the mixing network, a two-layer fully-connected neural network with ELU nonlinearity is used. 
To establish the monotonic relationship between $Q_{\texttt{tot}}$ and $Q^n$ (i.e., $\partial Q_{\texttt{tot}}/\partial Q^n \geq 0$), all weights of the mixing network are constrained to be non-negative. To this end, a hypernetwork is used to generate weights and biases for the mixing network, taking the global state as the input. Fully-connected layers with ReLU activation, followed by an absolute function, are employed to ensure non-negative weights. 
Since each agent network is conditional on its local observation and history, the learned policies can be extracted and executed independently without any negative impact on performance. We note that although having similar network architecture (except the down-sampling layer) as in the QMIX algorithm \cite{rashid2020monotonic}, our proposed asynchronous learning algorithm operates differently, especially in the way each agents takes its action. 

The training is end-to-end based on a replay buffer, aiming to minimize the loss function
\begin{equation}
    \mathcal{L}(\theta_m) = \mathbb{E}\left[\left(y_{\texttt{tot}} - Q_{\texttt{tot}}(s_m, \textbf{u}_m | \textbf{u}_m^{-\hat{n}}; \theta_m)\right)^2\right]
\end{equation}
where $y_{\texttt{tot}}$ is the target value estimated by one-step bootstrapping as
\begin{equation}
    y_{\texttt{tot}} = R(s_m, \textbf{u}_m) + \gamma^{\hat{\tau}_{m}}\underset{\textbf{u}_{m+1}}{\max}Q_{\texttt{tot}}(s_{m+1}, \textbf{u}_{m+1} | \textbf{u}^{-\hat{n}}_{m+1}; \theta^-_m) \notag
\end{equation}
wherein $\theta_m$ and $\theta^-_m$ is the parameters of the primary and the target networks.

\begin{algorithm}[h]
\caption{Asynchronous-QMIX}
\label{alg:aqmix}
\begin{algorithmic}[1]
\STATE Initialize the training environment, agents' policy and value networks, and buffer $\mathcal{D} \gets\varnothing$;
\STATE $m\gets 1$;
\STATE Observe $\mathbf{z}_1$ and sample action $\mathbf{u}_1$;
\STATE Initialize $\mathcal{Q} \gets \{(n, \tau^n_1): {n\in\mathcal{N}}\}$;\quad // $\tau^n_m$ is the time taken by agent $n$ to complete $u^n_m$
\WHILE{environment not terminal}
    \STATE $(\hat{n}, \hat{t}_{m+1}) \gets \arg\min_{(n, t)\in\mathcal{Q}}t$;
    \STATE Agent $\hat{n}$ observes reward $R(s_m, \mathbf{u}_m)$ and new observation $z^{\hat{n}}_{m+1}$;
    \STATE Agent $\hat{n}$ updates local history: $H^{\hat{n}}_{m + 1} \gets (H^{\hat{n}}_{m}, {u}^{\hat{n}}_{m}, z^{\hat{n}}_{m + 1})$;
    \STATE Agent $\hat{n}$ selects new action: $u^{\hat{n}}_{m+1} \gets \arg\max_{u\in\mathcal{A}^{\hat{n}}} Q^{\hat{n}}(H^{\hat{n}}_{m+1},u)$;
    \STATE Update joint action and observation: $\mathbf{u}_{m+1} \gets \{u^{\hat{n}}_{m+1}\}\cup \mathbf{u}^{-\hat{n}}_{m}$; $\mathbf{z}_{m+1} \leftarrow \{z^{\hat{n}}_{m+1}\}\cup \mathbf{z}^{-\hat{n}}_{m}$;
    \STATE Compute the new finish time:  $t_{\text{new}}\gets\hat{t}_m + \tau^{\hat{n}}_{m + 1}$;
    \STATE $\mathcal{Q}\gets\mathcal{Q}\cup\{(\hat{n},t_{\text{new}})\}$;
    \STATE $\mathcal{D}\gets \mathcal{D}\cup\{(s_m, \mathbf{z}_m, \mathbf{u}_m, R(s_m, \mathbf{u}_m), s_{m+1}, \mathbf{z}_{m+1})\}$;
    \STATE Update agents' policy and value networks;
    \STATE $m \gets m + 1$;
\ENDWHILE
\end{algorithmic}
\end{algorithm}

We name the proposed framework Asynchronous-QMIX (AQMIX), which is outlined in Algorithm \ref{alg:aqmix}. Since different actions of different agents are completed at different timestamps, the idea is to maintain a list $\mathcal{Q}$ that records when each agent finishes its action. At each training step $m$, only the agent with the smallest associated timestamp is processed, i.e., its observation is updated and a new action is generated. This asynchrony is the key distinction between the proposed algorithm and the canonical synchronous QMIX.

\subsection{Enhance Scalability through State Downsampling}
The hypernetwork in the mixing network requires the positions of UAVs and their completion maps to generate the mixing parameters benefiting the UAVs cooperation. However, the incorporation of these maps into the state significantly expands the state space. With $N$ UAVs and a monitored area of $H^2$ cells, incorporating the completion maps exponentially expands the state space by $2^{NH^2}$. In fact, when the numbers of cells and UAVs grow, the state space dimension is dominated by the size of these maps. This expansion not only increases training time but also makes the model harder to train due to the curse of dimensionality, potentially degrading learning performance (as we will show in the experiment later). Moreover, including the completion status of each UAV at every cell may be redundant. This is because, from the global mixing network perspective, knowing UAVs' positions and their completion status in each collecting region (i.e., group of cells in close proximity) is sufficient for approximately evaluating a cooperative strategy among UAVs.

Motivated by the observations above, we employ a sum-pooling layer to downsample the resolution of completion maps before feeding them into the hypernetwork of the mixing network. With a kernel size of $M\times M$ and non-overlapping sampling windows, this layer can reduce the state space by a factor of $2^{M^2}$, thereby reducing the size of the mixing network and the training time. The downsampling layer is also illustrated Fig.~\ref{fig:aqmix} as the first block of the hypernetwork.

\subsection{Computational Complexity in Deployment}
Since only agent networks are extracted and deployed on UAVs after training, the complexity of calculating actions in deployment does not depend on the mixing network. The input of each agent network includes UAV's position, energy level, data collection demands of observable cells, and the completion map $G^n \in \{0, 1\}^{H^2\times 1}$. Since the number of observable cells does not exceed the size of completion map $G^n$, the input size is dominated by $H^2$. As illustrated in Fig.~\ref{fig:aqmix}, each agent network comprises one fully-connected input layer, followed by a GRU layer, and finally, a fully-connected output layer. Given that the complexity of a fully-connected layer with input size $N_{\text{x}}$ and hidden size $N_{\text{h}}$ is $\mathcal{O}(N_{\text{x}}N_{\text{h}})$, the complexity of a GRU cell with the same input size $N_{\text{x}}$ and hidden size $N_{\text{h}}$ is $\mathcal{O}(N_{\text{x}}^2 + N_{\text{x}}N_{\text{h}})$ \cite{cho2014learning}, and the fact that we use the same number of $N_{\text{h}}$ units for all hidden layers of the policy networks, the complexity to calculate an action for a UAV at each step is given by $\mathcal{O}(H^2N_{\text{h}} + N_{\text{h}}^2)$.

\section{Bandwidth allocation under imperfect CSI}
\label{sec:bandwidth_optimization}
Once the UAVs have learned the trajectory policy, they fly through each cell and hover to collect data using the FDMA protocol. To minimize the hovering time of the UAV, this section optimizes the bandwidth allocation to active SNs considering the practical imperfect CSI. In practice, perfect CSI is infeasible, resulting channel estimation error. Let $\bs{h}^{in}(t), \hat{\bs{h}}^{in}(t)$ be the true and estimated channel coefficients between UAV $n$ and SN $i$, respectively. Under imperfect CSI, we have $\bs{h}^{in}(t)  = \hat{\bs{h}}^{in}(t) + \bs{e}^{in}(t)$, where $\bs{e}^{in}(t)$ is the estimation error that is statistically independent from the estimated channel and its elements are modelled as random variables with zero mean and variance $\sigma^2_e/N_{tx}$. 
Under FDMA transmission mode, there is no interference at the receiver side. The received signal from SN $i$ under MRC receiver at UAV $n$ is given as
\begin{equation}
    \hat{y}^{in}(t) = \sqrt{P_s} \frac{\tilde{\bs{h}}^H}{\|\tilde{\bs{h}}\|}\big(\hat{\bs{h}}^{in}(t)+ \bs{e}^{in}(t) \big)x^i(t) + \frac{\tilde{\bs{h}}^H}{\|\tilde{\bs{h}}\|} \bs{n}^{in}(t)
\end{equation}
where $\tilde{\bs{h}} \triangleq \hat{\bs{h}}^{in}(t)$, $x^i(t)$ is the transmitted symbol with unit average power over the symbol constellation and $\bs{n}^{in}(t)$ is the thermal Gaussian noise. By treating the channel estimation error as noise, the achievable rate for SN $i$ is given by
\begin{equation}
    \label{eq:data_rate_imp_csi}
    \bar{R}^{in}(t) = b^{in}(t)\log_2\Big(1 + \frac{{\|\hat{\bs{h}}^{in}(t)\|}^2 P_s}{P_s\sigma_e^2 + b^{in}(t)N_0}\Big).
\end{equation}

Let $\mathcal{I}_n(t)$ be the set of active SNs of the cell serving by UAV $n$ at time $t$. 
The bandwidth allocation optimization problem for the $n$-th UAV at the time $t$ to minimize the hovering time can be formulated as follows:
\begin{align}
    \min_{\{b^{in}(t) \geq 0\}} \ \max_{i\in\mathcal{I}_n(t)}\frac{D^i}{\bar{R}^{in}(t)};\  \text{s.t.}\ {\sum}_{i\in \mathcal{I}_n(t)} b^{in}(t) \leq B. \label{prob:bandwidth}
\end{align}
To solve problem \eqref{prob:bandwidth}, we introduce an auxiliary variable $\zeta \geq 0$ as the the total hovering time of UAV $n$. The hovering time needs to guarantee that all data is collected, i.e.,
\begin{equation}
    \bar{R}^{in}(t) \geq D^i/ \zeta, \forall i\in \mathcal{I}_n(t). \label{eq:const29}
\end{equation}
These non-convex constraints can be convexified by dividing both sides by a positive $\zeta$. Thus, problem  \eqref{prob:bandwidth} can be reformulated as as follows:
\begin{align}
    \min_{\{b^{in}(t) \geq 0\}, \zeta \geq 0}  \zeta; \,
    \text{s.t.}~ \eqref{eq:const29} \text{ and } {\sum}_{i\in \mathcal{I}_n(t)} b^{in}(t) \leq B. \label{prob:bandwidth_reformulated}
\end{align}
It can be shown that the rate function in \eqref{eq:data_rate_imp_csi} is a concave function with respect to variable $b^{in}(t)$, though the proof is omitted due to space limitation. Thus, problem \eqref{prob:bandwidth_reformulated} is a convex optimization problem with a linear objective function and convex constraints, and can be efficiently solved by standard methods, e.g., interior point. 
Since this problem only depends on the information at the current cell $c$, it can be solved to optimality  based on local observation of the UAV at each hovering point.

\section{Simulation Results and Discussions}
\label{sec:simulation}
In this section, we evaluate the performance of the proposed method in solving problem~\eqref{prob:P}.
We first describe the simulation setups and baseline algorithms. Performance comparisons and analyses will be discussed subsequently.

\subsection{Simulation Setups} \label{sec: simulation setup}
The monitored area is divided into a grid of cells, each with a size of 50m$\times$50m. To evaluate the scalability of the algorithms, we use different grid sizes, including 8$\times$8, 10$\times$10, 15$\times$15, and 20$\times$20 cells. Additionally, we assess the algorithms with varying numbers of UAVs, ranging from 1 to 6. Since the state space grows exponentially with the number of cells, and adding more UAVs directly amplifies the nonstationarity of the environment, larger grid sizes and more UAVs make it increasingly difficult for the algorithms to learn and make effective decisions. As sensors are often deployed non-uniformly in practice, with a higher concentration around targets, we randomly divide the cells into two groups, including sparse cells and dense cells with ratio of 7:3, respectively.
The number of SNs in each cell is generated using Poisson distribution Pois($\lambda$), where $\lambda = 10$ for dense cells and $\lambda = 1$ for sparse cells. The SNs are then uniformly placed within each cell, as illustrated in Fig.~\ref{fig:topology}. To generate data availability for collection, we first generate the number of cells containing data following Poisson distribution with the mean value $\lambda = \phi H^2$, where $\phi$ is used to control the density of data demands. Once the number of cells containing data is determined, their locations are randomly assigned. The data size of active SNs is randomly generated from $[0.1, 1.0]$ Mbits, following some stable distribution. We set the transmit power of SNs $P=10$ dBm, attenuation due to NLoS $\beta = 0.2$, 
total bandwidth $B=1$ MHz, noise power spectral density $N_0=-150$ dBm, and the pathloss exponent $\eta=2.6$.
The UAVs are assumed to fly at $h=100$ m altitude.
For the environment parameters, we set $a = 11.95, b=0.14, A_1 = 1.0,$ and $A_2 = 4.39$, following the settings in \cite{wang2023cooperative, you20193d}.
Initial locations and final destinations of all UAVs are set at the same position, at the center of the bottom-left cell depicted in Fig.~\ref{fig:topology}. 
The observable region of each UAV is an area of 3$\times$3 cells centered at its location. All UAVs are trained with energy budgets of $E^n_{\max} = 1000$ kJ each. Unless otherwise indicated, the following default parameters are used. There are $N = 2$ UAVs with different velocities of 5 m/s and 10 m/s, data density $\phi = 0.3$, and inter-UAV communication range is set to 200 m. All other UAVs' parameters are retained as in~\cite{zeng2019energy}.

\begin{figure}[t]
     \centering
     \begin{subfigure}{\columnwidth}
         \centering
         \includegraphics[width=\textwidth]{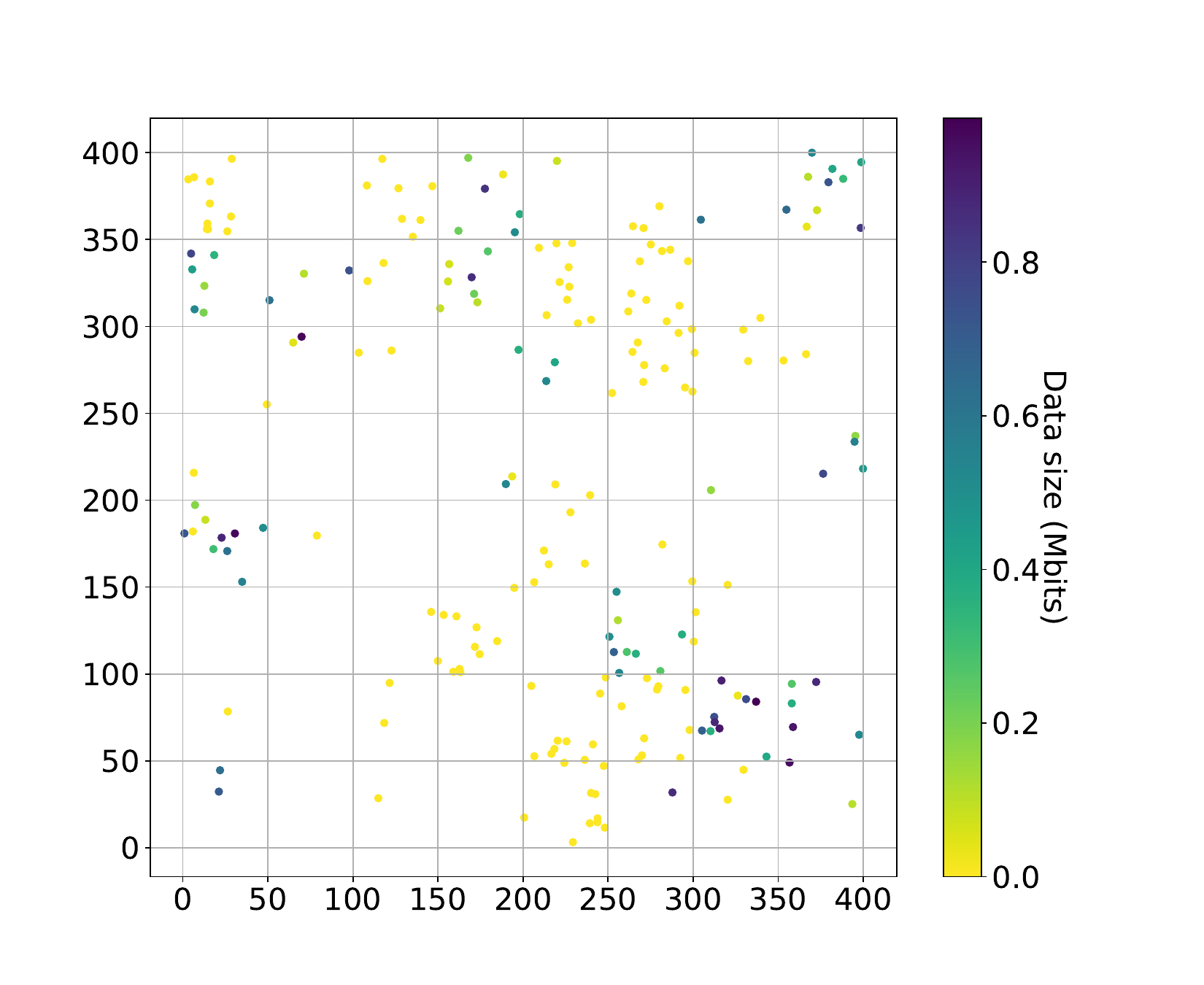}
     \end{subfigure}
     \vspace{-0.5cm}
     \caption{Distribution of SNs over the 8$\times$8 cell collecting area. The colors of SNs represent an example of data collection demands with available data size ranging from 0 to 1 Mbits.}
     \label{fig:topology}
     \vspace{-0.5cm}
\end{figure}

The proposed AQMIX algorithm is compared with following reference schemes:
\begin{itemize}
    \item Independent learning (AIQL): this is a fully decentralized baseline obtained by removing the mixing network from the architecture of AQMIX, which is analogous to the paradigm used in \cite{bayerlein2021multi}. The policies are learned only based on local action-observation histories without using the global state. 
    \item QMIX \cite{wang2023cooperative}: this algorithm was designed only for synchronous environments. To adapt it in our asynchronous setting, we employ a synchronous training - asynchronous deployment mechanism as follows. During training, if a UAV completes its action before others, it waits until all UAVs complete their actions, allowing them to calculate their new actions simultaneously. Since all learned policies can be extracted and executed distributively, we then deploy them in our asynchronous environment without inter-UAV synchronization. This baseline enables us to evaluate both the learning performance of the proposed algorithm relative to the original QMIX, as well as the performance of the existing synchronous MARL algorithm in asynchronous environments.
    \item Heuristic (HERT): a very naive but feasible solution to~\eqref{prob:P} is to partition the area into multiple sectors and assign each part to one UAV. UAVs then fly over their assigned sub-areas, exploring and collecting data cell by cell. 
\end{itemize}

To prevent algorithms from learning trivial solutions by memorizing specific trajectories, the data collection demands are randomly regenerated at the beginning of each episode, encouraging agents to learn more general behaviors.
The hyperparameters are hand-tuned for reasonable performance and are used across all methods. Specifically, we use a learning rate of 5e-5, a discount factor of $\gamma = 0.99$, a batch size of 32, a replay buffer size of $1e6$ samples, and a target network update rate of 1e-2. We set the number of neurons to 256 in all hidden layers of the policies and the mixing network. The state downsampling is performed with kernel size of 3$\times$3. The reward function is scaled down by the maximum generated data size (which is 1.0 Mbits) to avoid numerical issues. Besides, we adopt the equal bandwidth allocation  during training to minimize the training time, and only perform bandwidth optimization during the testing phase.

\subsection{Performance Comparisons and Analyses}
\label{sec:result - performance comparison}
\subsubsection{Learning performance}
\begin{figure*}[ht]
     \centering
     \begin{subfigure}{0.32\textwidth}
         \centering
         \includegraphics[width=\textwidth]{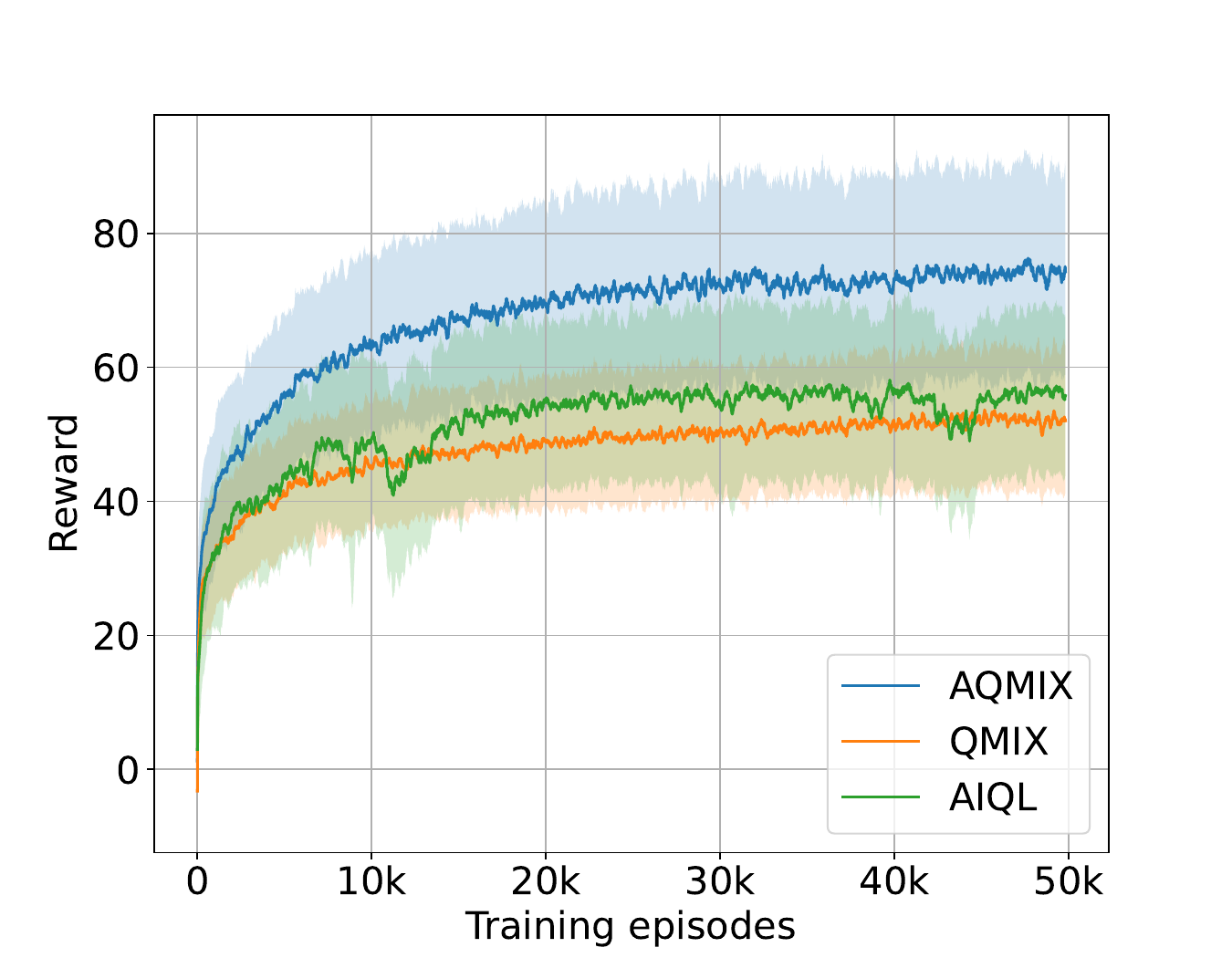}
         \caption{10$\times$10 cells}
     \end{subfigure}
     \hfill
     \begin{subfigure}{0.32\textwidth}
         \centering
         \includegraphics[width=\textwidth]{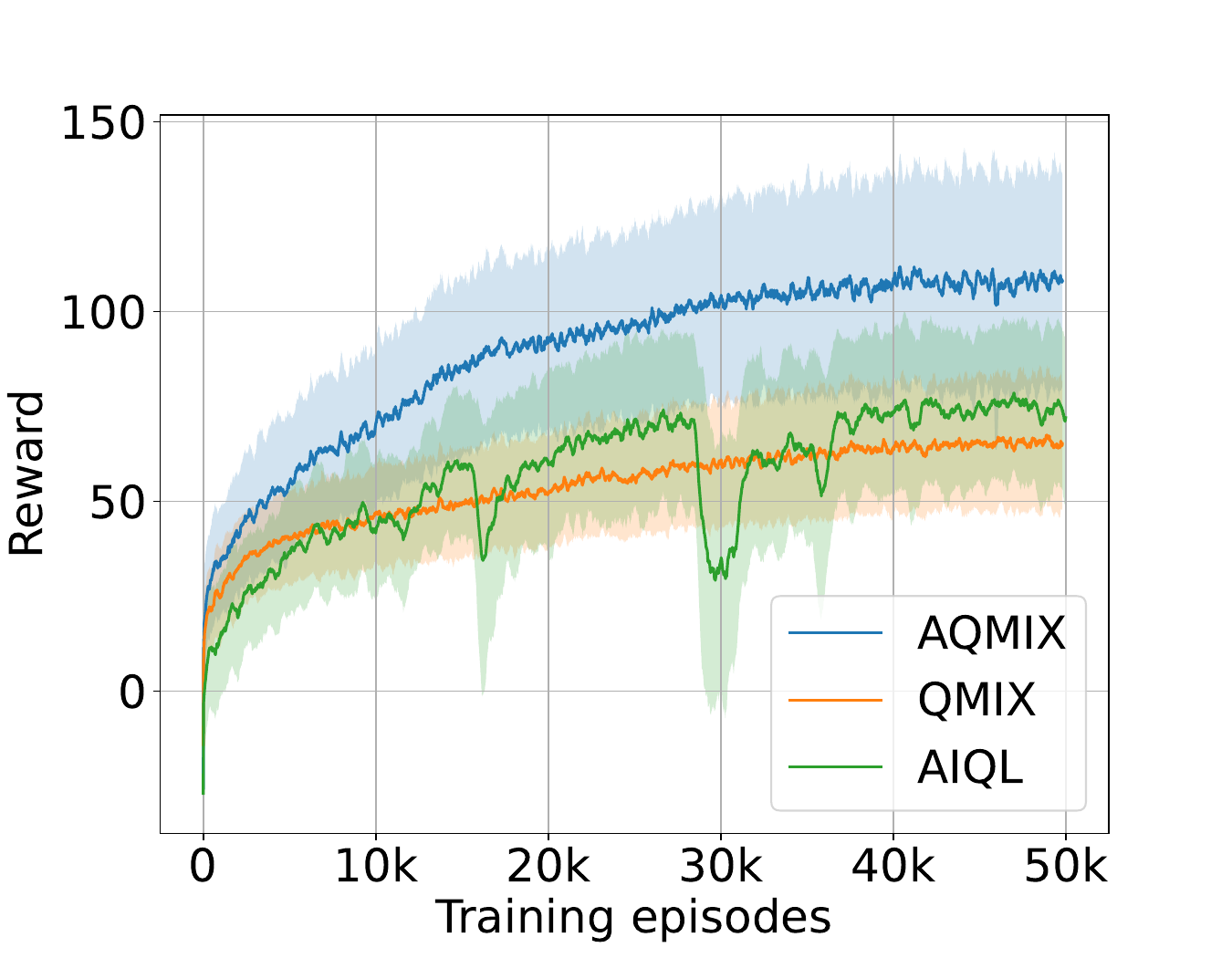}
         \caption{15$\times$15 cells}
     \end{subfigure}
     \hfill
     \begin{subfigure}{0.32\textwidth}
         \centering
         \includegraphics[width=\textwidth]{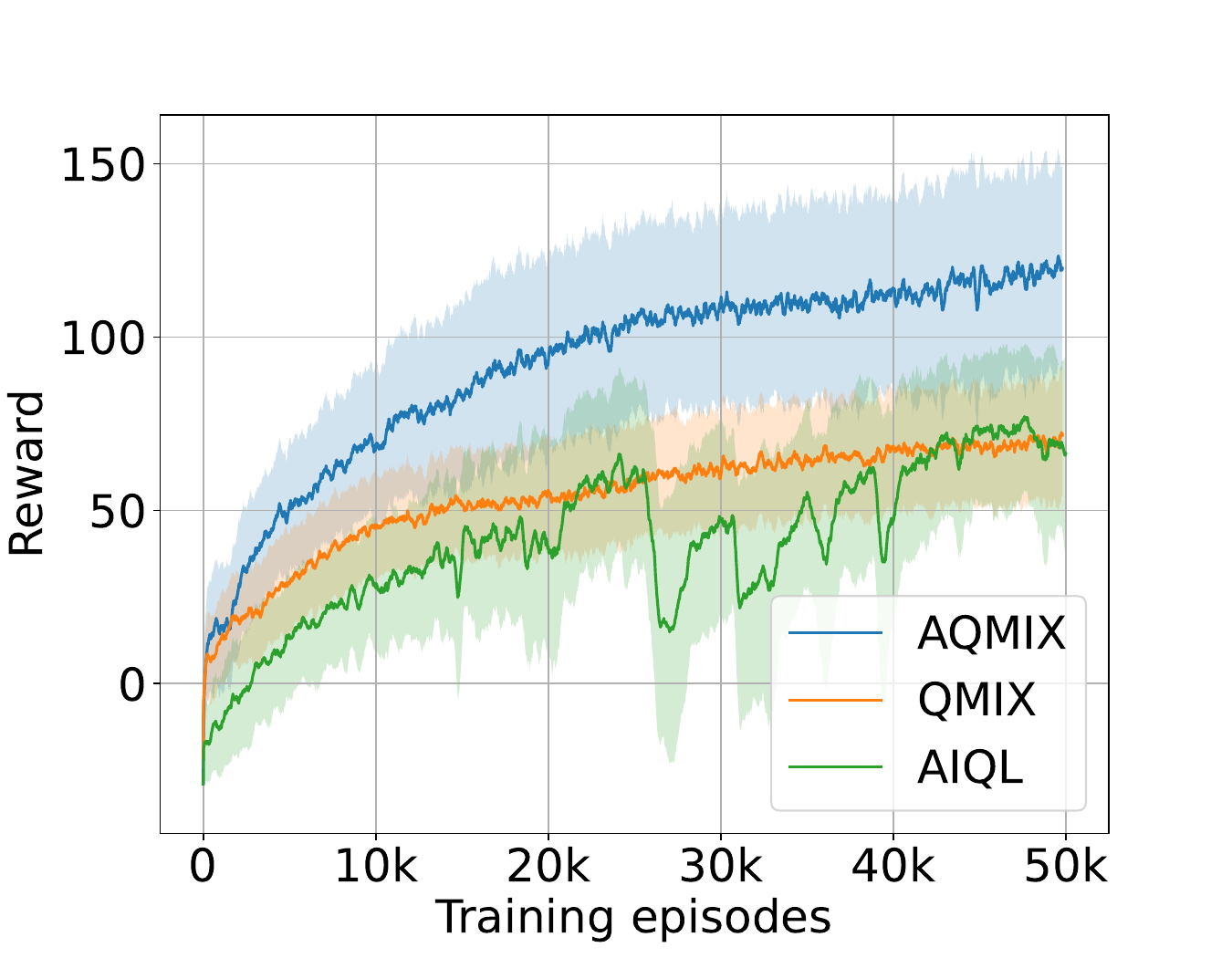}
         \caption{20$\times$20 cells}
     \end{subfigure}
     \caption{Total reward per episode during training on different network sizes.}
     \label{fig:learning_curve_netsize}
     \vspace{-0.4cm}
\end{figure*}

\begin{figure*}[ht]
     \centering
     \begin{subfigure}{0.24\textwidth}
         \centering
         \includegraphics[width=\textwidth]{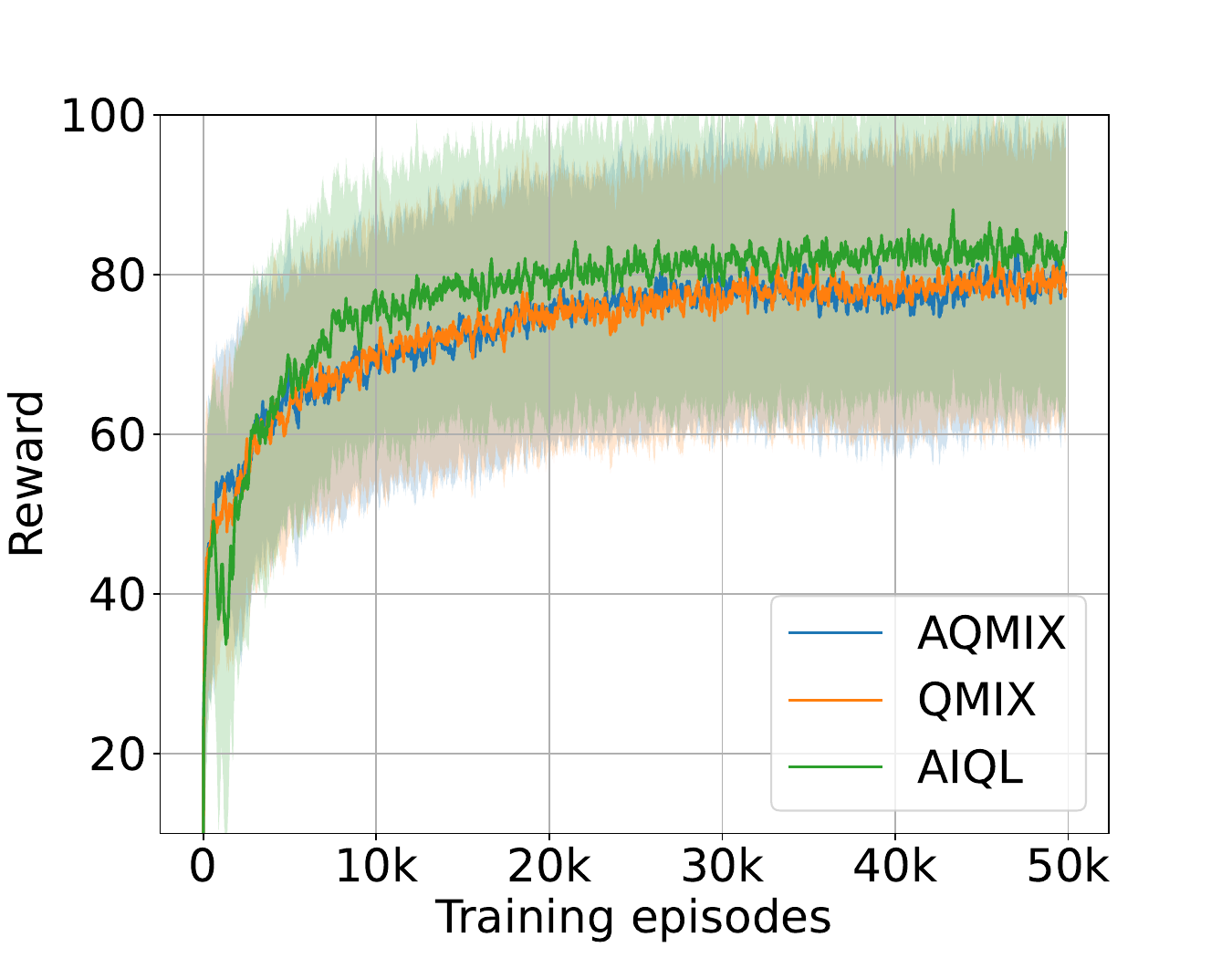}
         \caption{1 UAV}
     \end{subfigure}
     \hfill
     \begin{subfigure}{0.24\textwidth}
         \centering
         \includegraphics[width=\textwidth]{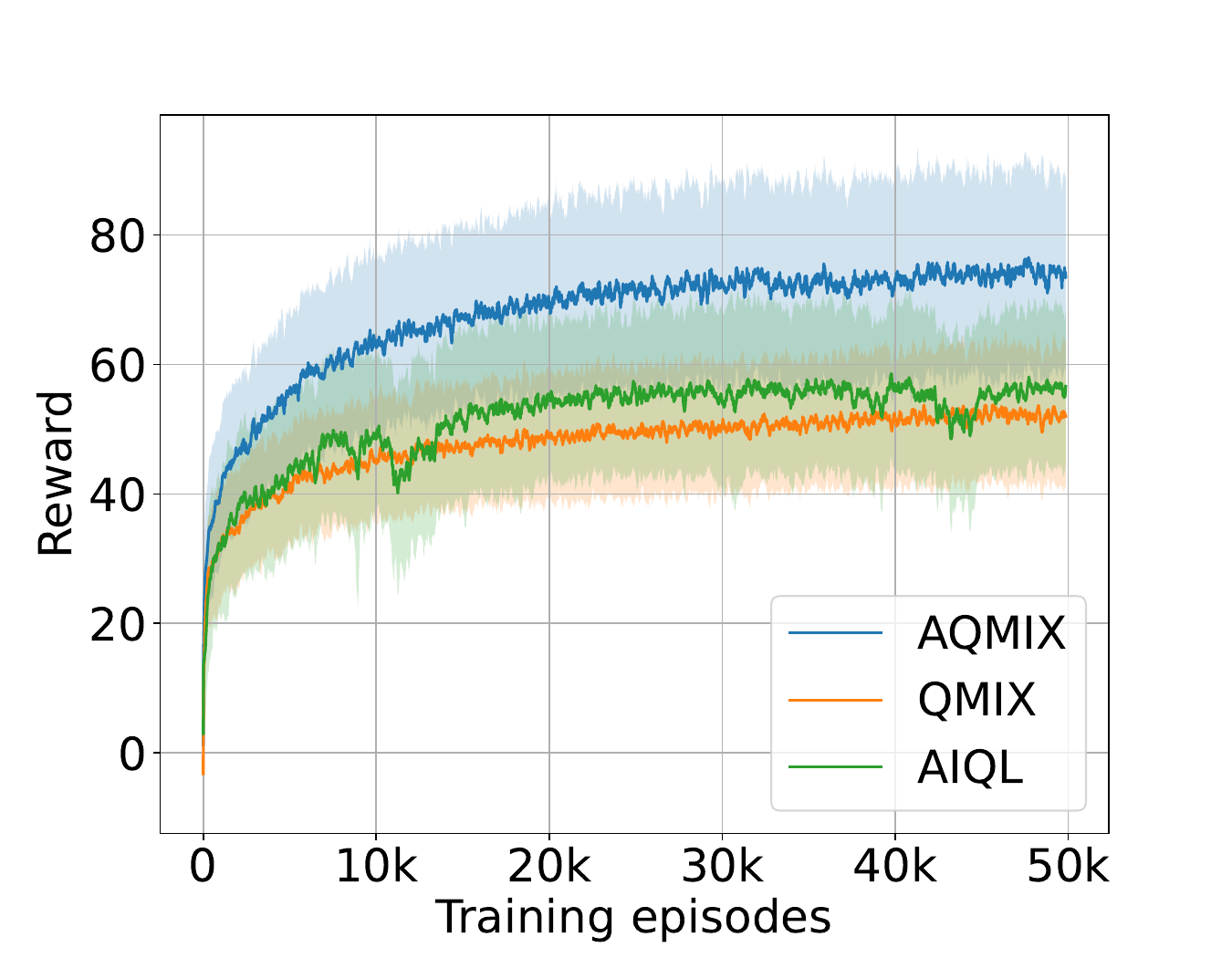}
         \caption{2 UAVs}
     \end{subfigure}
     \hfill
     \begin{subfigure}{0.24\textwidth}
         \centering
         \includegraphics[width=\textwidth]{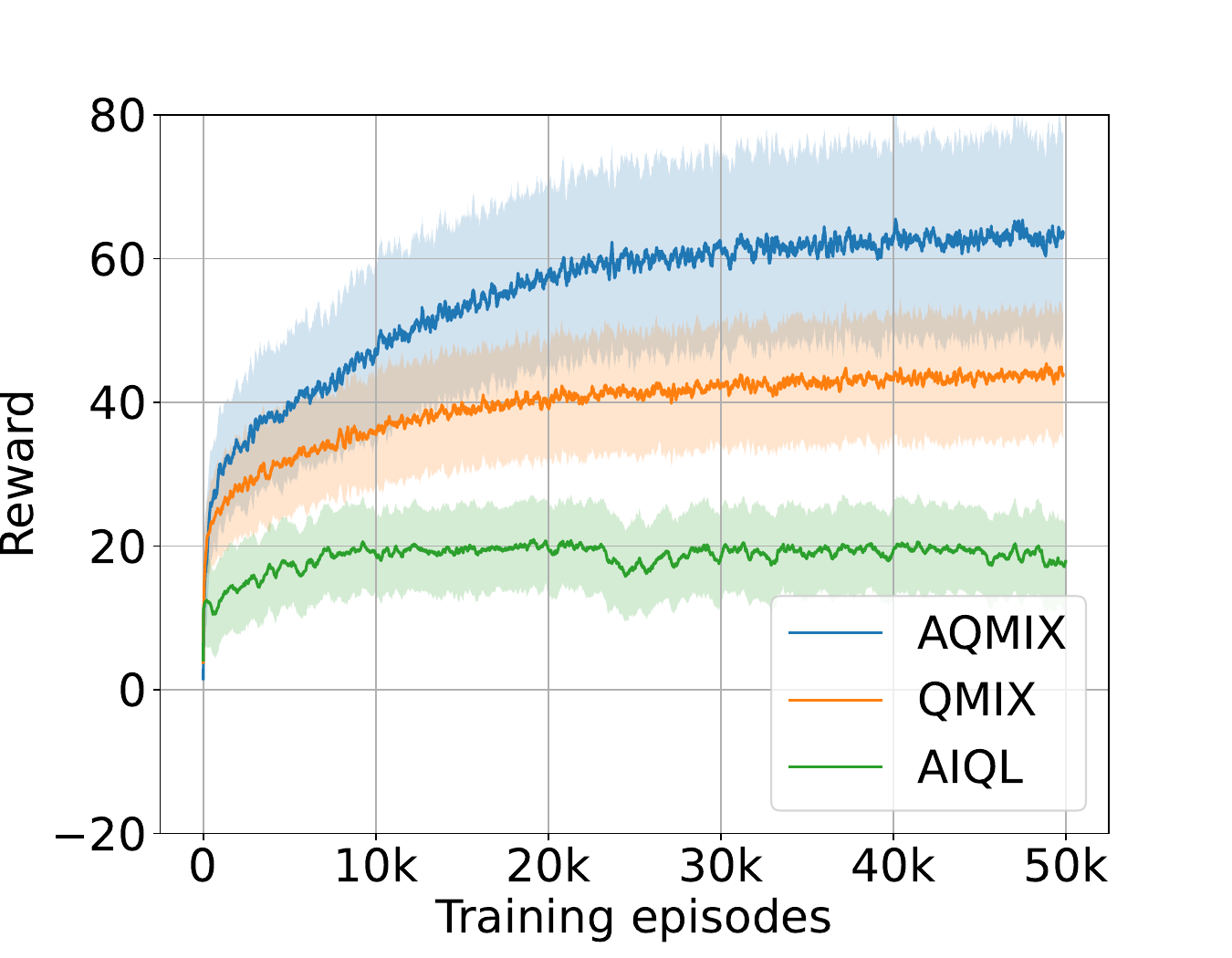}
         \caption{4 UAVs}
     \end{subfigure}
     \hfill
     \begin{subfigure}{0.24\textwidth}
         \centering
         \includegraphics[width=\textwidth]{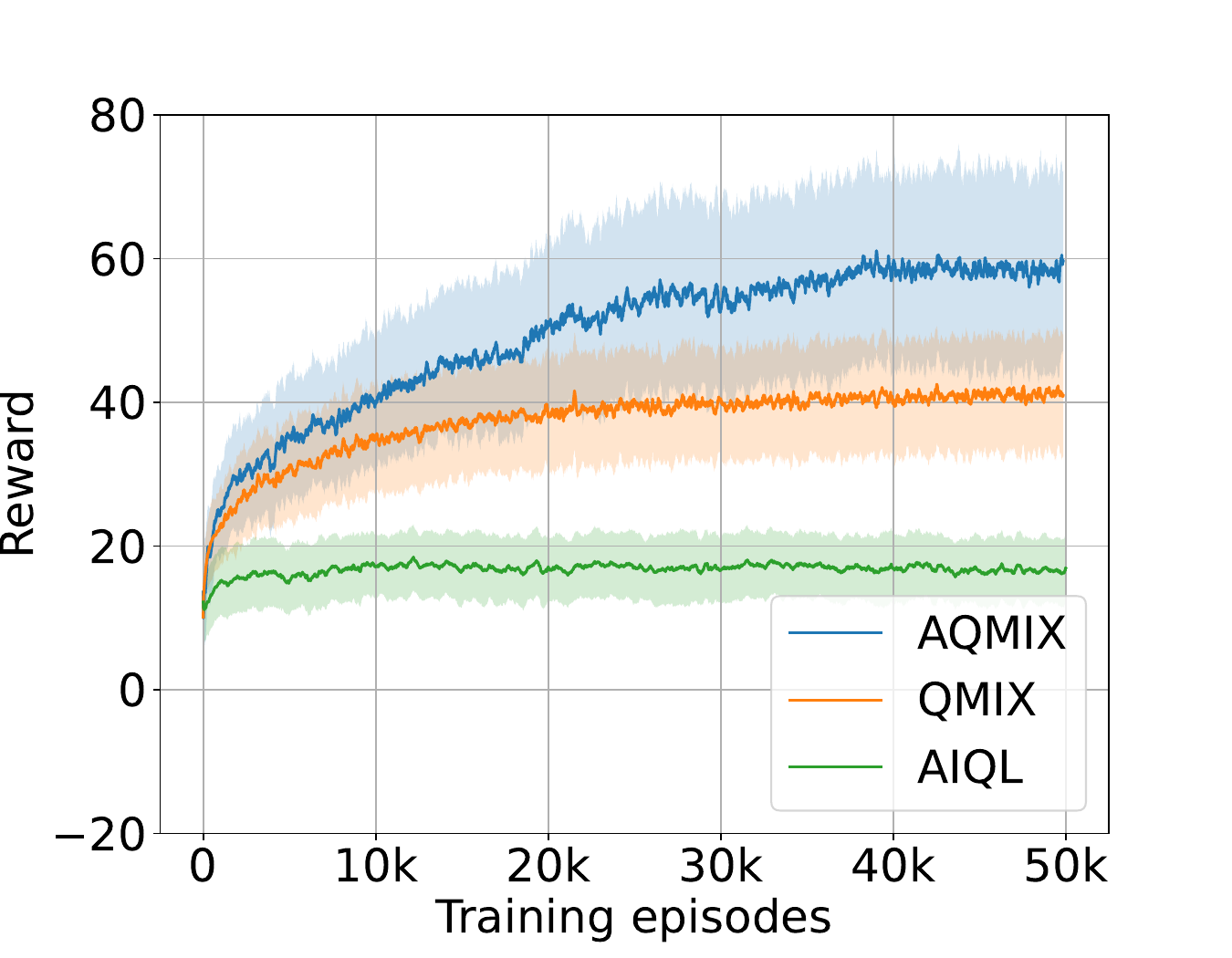}
         \caption{6 UAVs}
     \end{subfigure}
     \caption{Total reward per episode during training with different numbers of UAVs.}
     \label{fig:learning_curve_nb_uavs}
     \vspace{-0.5cm}
\end{figure*}
We first examine the learning performance of all learning-based solutions on different numbers of cells. Fig.~\ref{fig:learning_curve_netsize} plots the total reward per episode during training, where we train two UAVs to collect data over areas of 10$\times$10, 15$\times$15, and 20$\times$20 cells. In this figure, each line represents the average values over 10 different trainings and the shaded areas represent the standard deviation. It can be seen that AQMIX clearly outperforms QMIX and AIQL across all network sizes. The superiority of AQMIX over QMIX can be explained by the fact that QMIX requires UAVs to wait for synchronization with others before making decisions, resulting in inefficient use of hovering energy. These results demonstrate that we have successfully extended QMIX, preserving its advantages in an asynchronous environment. While AQMIX and QMIX exhibit consistent convergence trends during learning, the results of AIQL show significant variance in learning performance. These performance fluctuations can be attributed to the absence of cooperation mechanisms between agents, which is a common limitation of independent learning methods.

Fig.~\ref{fig:learning_curve_nb_uavs} plots the total reward per episode during training, where we train $1, 2, 4$, and $6$ UAVs to collect data over areas of 10$\times$10 cells. As expected, sub-fig.~\ref{fig:learning_curve_nb_uavs}a shows that when there is only one UAV, all algorithms exhibit the same learning performance. This is because, in this case, asynchronous and synchronous environments are identical, and all learning algorithms function as single-agent RL algorithms. As the number of UAVs increases, the gaps between algorithms become clearer, with AQMIX standing out as the best solution. Notably, the performance of AIQL significantly drops as the number of UAVs increases. This is because a greater number of UAVs leads to a higher likelihood of inter-UAV communication, i.e., more frequent synchronization of information among UAVs. In the case of AQMIX and QMIX, this synchronization is beneficial due to their cooperative mechanisms. However, due to the absence of cooperation management mechanisms in AIQL, these synchronizations inadvertently trigger unexpected interactions between UAVs, exacerbating the inherent instability of the algorithm.

Overall, AQMIX has proven to be the most effective learning solution in asynchronous environments, as demonstrated by its stable and highly competitive performance. In contrast, AIQL fails to leverage the benefits of increasing the number of UAVs, as well as the potential of inter-UAV communication.

\subsubsection{Robustness of learned policies}
to evaluate the robustness of learned policies, we select the best policy obtained by each method\footnote{For the three learning algorithms, evaluating a massive number of policies generated during training is too costly, and thus we only store and evaluate a policy after every 100 episodes. Each of these policies is tested on thirty different scenarios of data collection demands, and the best policy is determined by its highest energy efficiency.}, and then evaluate these policies on a same set of 1000 different scenarios of data generations. In the next two figures, we report the testing results on the test benchmark of 10$\times$10 cells.

\begin{figure*}[ht]
     \centering
     \begin{subfigure}{0.32\textwidth}
         \centering
         \includegraphics[width=\textwidth]{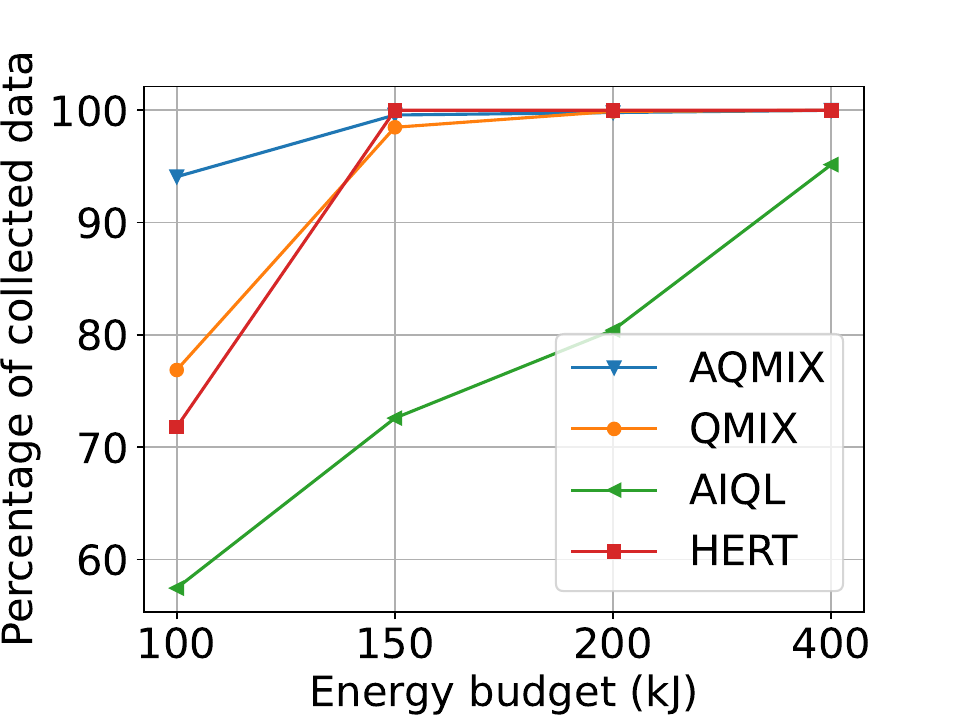}
         \caption{Percentage of collected data}
     \end{subfigure}
     \hfill
     \begin{subfigure}{0.32\textwidth}
         \centering
         \includegraphics[width=\textwidth]{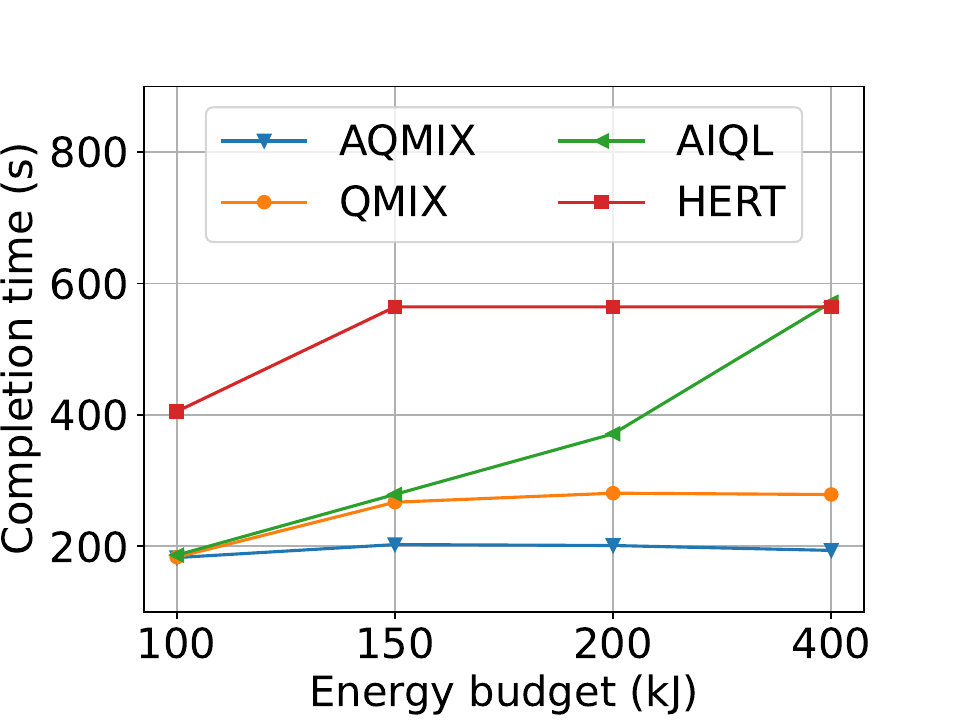}
         \caption{Completion time}
     \end{subfigure}
     \hfill
     \begin{subfigure}{0.32\textwidth}
         \centering
         \includegraphics[width=\textwidth]{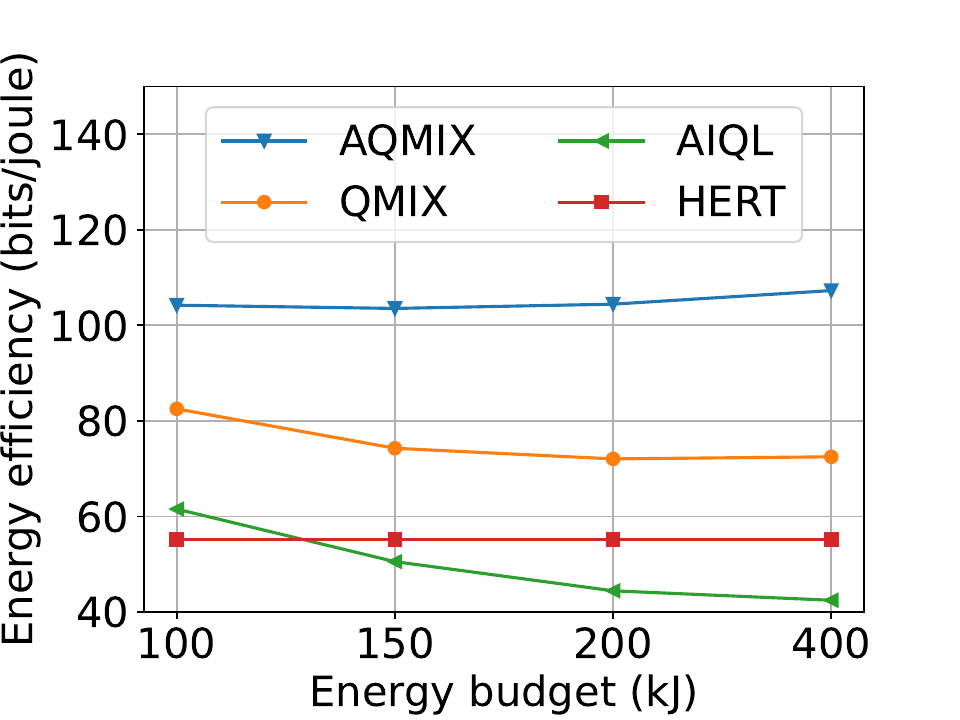}
         \caption{Energy efficiency}
     \end{subfigure}
     \caption{Testing results on different levels of energy budget.}
     \label{fig:energy_robustness}
     \vspace{-0.5cm}
\end{figure*}

Fig.~\ref{fig:energy_robustness} plots the average results over 1000 testing scenarios, with the energy budget for each UAV varying from 100kJ to 400kJ. The methods are evaluated under three criteria: the percentage of collected data (Fig.~\ref{fig:energy_robustness}a), the mission completion time (Fig.~\ref{fig:energy_robustness}b), and energy efficiency (Fig.~\ref{fig:energy_robustness}c). It’s worth noting that all tested policies were trained only with $E^n_{\max} = 1000$kJ.
Sub-fig.~\ref{fig:energy_robustness}a shows that 150kJ is sufficient for AQMIX, QMIX, and HERT to collect all data, while this number for AIQL is above 400kJ.
Having the energy budget exceed this required level does not have a significant impact on the solution quality, as shown in sub-fig.~\ref{fig:energy_robustness}b and sub-fig.~\ref{fig:energy_robustness}c, where the completion time and energy efficiency remain unchanged as the energy budget increases.
Overall, the proposed AQMIX algorithm provides outstanding performance in all three criteria compared to the benchmark schemes, demonstrated by its higher energy efficiency and the ability to collect the most data with the least time and energy consumption. The results also suggest that mission completion time and total collected data can be optimized indirectly by maximizing energy efficiency.

\begin{figure}[t]
     \centering
     \begin{subfigure}{0.49\columnwidth}
         \centering
         \includegraphics[width=\textwidth]{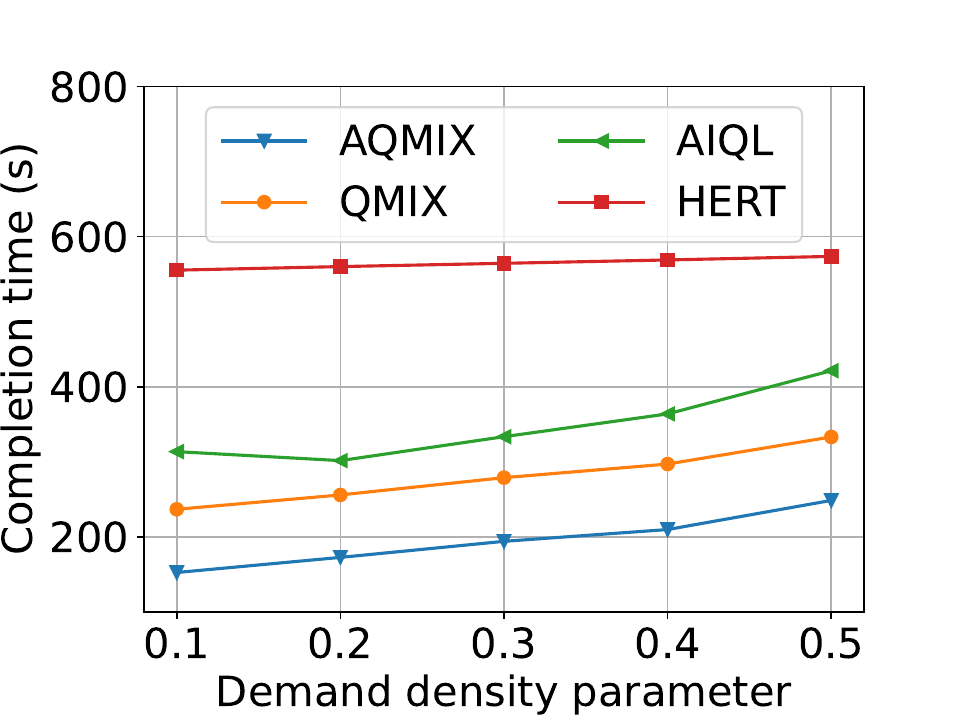}
         \caption{Completion time}
     \end{subfigure}
     \hfill
     \begin{subfigure}{0.49\columnwidth}
         \centering
         \includegraphics[width=\textwidth]{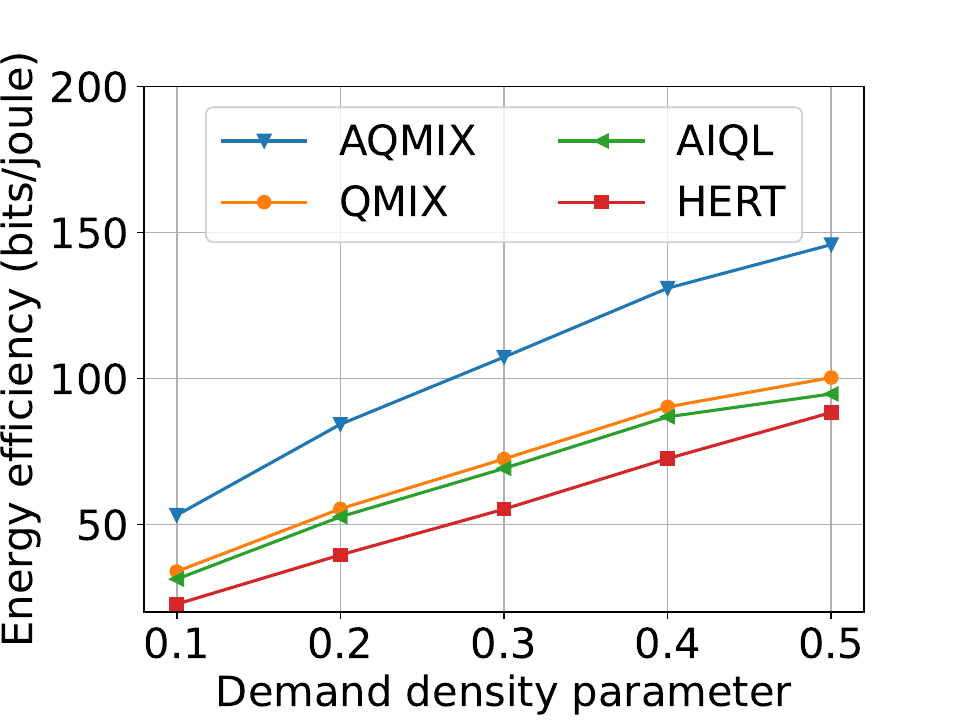}
         \caption{Energy efficiency}
     \end{subfigure}
     \caption{Testing result on different density levels of data collection demands.}
     \label{fig:demand_density}
     \vspace{-0.4cm}
\end{figure}
Fig.~\ref{fig:demand_density} shows the average completion time and energy efficiency achieved by all schemes, where the data density parameter $\phi$ is varied between 0.1 and 0.5. It’s also worth noting that all tested policies were trained only at $\phi = 0.3$. As shown in Fig.~\ref{fig:demand_density}, both the completion time (sub-fig.~\ref{fig:demand_density}a) and energy efficiency (sub-fig.~\ref{fig:demand_density}b) tend to increase as the amount of data to collect increases. This is because, as more data is available in the same area, the energy cost per unit of data collected decreases. This leads to improved energy efficiency, as the UAVs can collect more data without a proportional increase in energy consumption. Similar to the previous experiment, AQMIX consistently outperforms the benchmark schemes.

In summary, Fig.~\ref{fig:energy_robustness} and Fig.~\ref{fig:demand_density} demonstrate that the proposed AQMIX algorithm not only outperforms other reference schemes, but also exhibits robustness and strong generalization ability across varying conditions of data availability and UAV energy budgets.

\subsubsection{Impacts of the number of UAVs} 
\begin{figure}[t]
     \centering
     \begin{subfigure}{0.49\columnwidth}
         \centering
         \includegraphics[width=\textwidth]{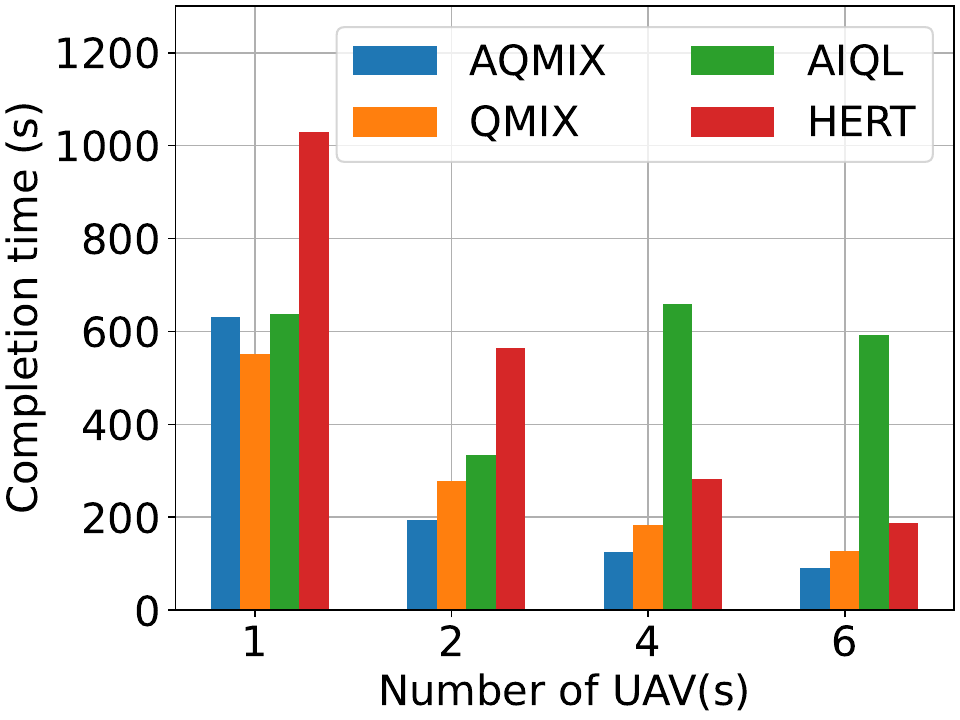}
         \caption{Completion time}
     \end{subfigure}
     \hfill
     \begin{subfigure}{0.49\columnwidth}
         \centering
         \includegraphics[width=\textwidth]{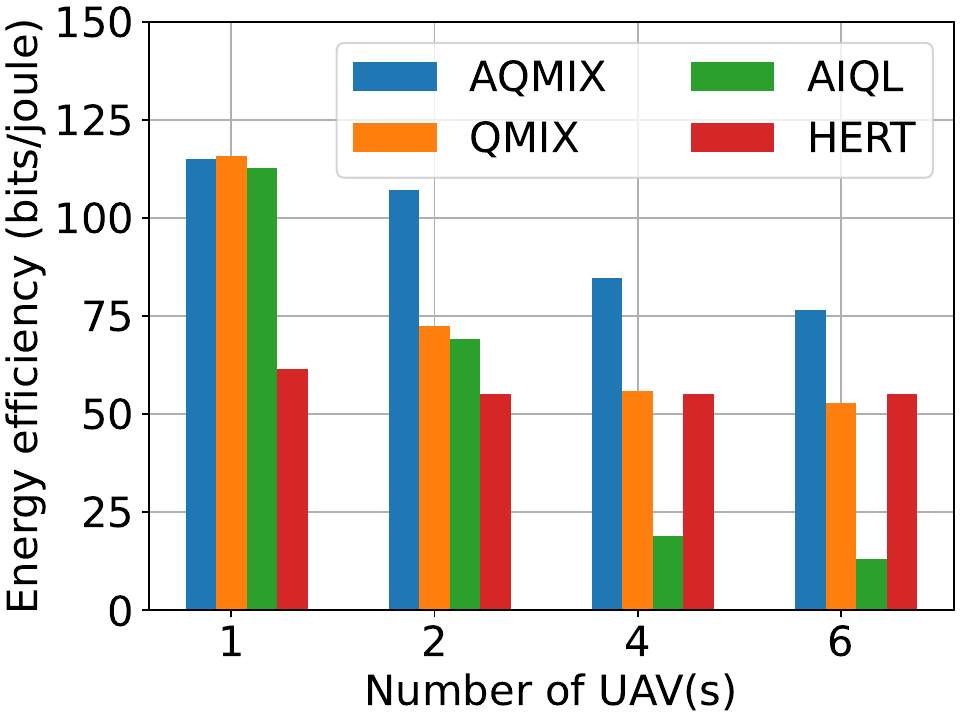}
         \caption{Energy efficiency}
     \end{subfigure}
     \caption{Testing results of policies trained with different numbers of UAVs.}
     \label{fig:nb_uavs}
\end{figure}
Fig.~\ref{fig:nb_uavs} plots the testing results where policies are trained with different numbers of UAVs to collect data over areas of 10$\times$10 cells.
In this figure, the results of AQMIX and QMIX reveal a trade-off between completion time and energy efficiency. In particular, increasing the number of UAVs significantly reduces the completion time (sub-fig. \ref{fig:nb_uavs}a) but also decreases the energy efficiency (sub-fig. \ref{fig:nb_uavs}b). While the reduction in completion time is straightforward, the decline in energy efficiency can be attributed to the overlapping UAV trajectories as the number of UAVs increases. Additionally, having more UAVs makes the environment more challenging to learn due to increased noise and uncertainty in the UAVs' observations.
Notably, increasing the number of UAVs from one to two reduces the completion time by approximately 75\% in AQMIX, while efficiency decreases only slightly, by less than 10\%. This trade-off ratio is significantly better than that of QMIX, which sacrifices 30\% of energy efficiency for a 50\% reduction in completion time.
In the case of HERT, the completion time decreases proportionally with the number of UAVs, while the energy efficiency remains unchanged. This is because each UAV is assigned a specific, non-overlapping area, meaning the total energy consumption by all UAVs remains roughly the same as in the single-UAV case. Finally, AIQL demonstrates poor performance when the number of UAVs increases, since this algorithm struggles to learn when there are more than two UAVs, as discussed previously.
Overall, AQMIX achieves the lowest completion time and the highest energy efficiency in all cases.

\subsubsection{Impacts of inter-UAV communication}

\begin{figure}[ht]
     \centering
     \begin{subfigure}{0.49\columnwidth}
         \centering
         \includegraphics[width=\textwidth]{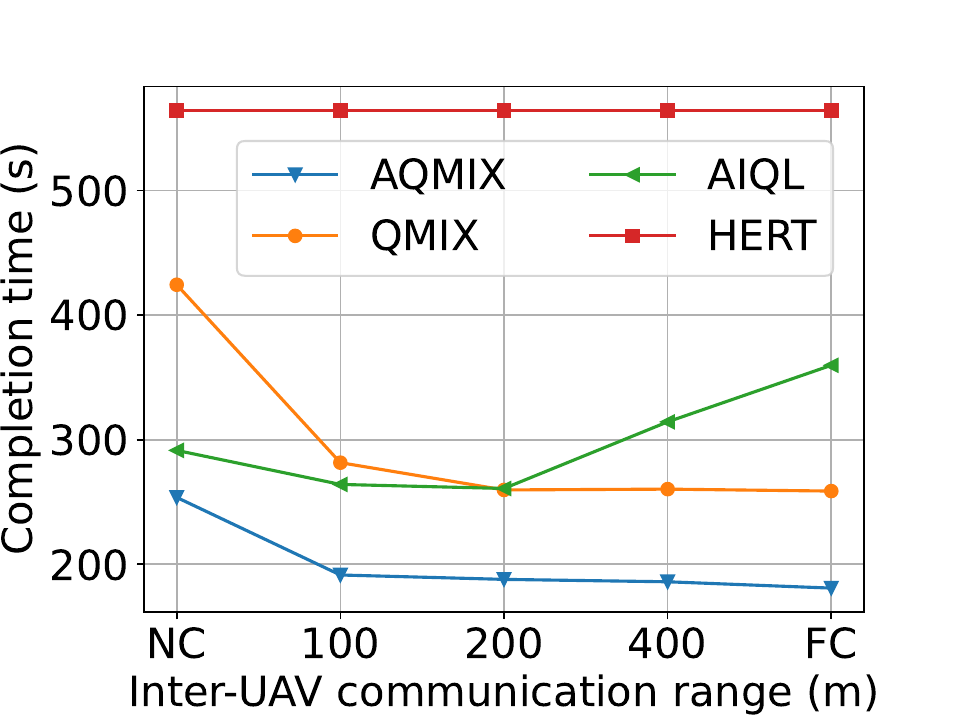}
         \caption{}
     \end{subfigure}
     \hfill
     \begin{subfigure}{0.49\columnwidth}
         \centering
         \includegraphics[width=\textwidth]{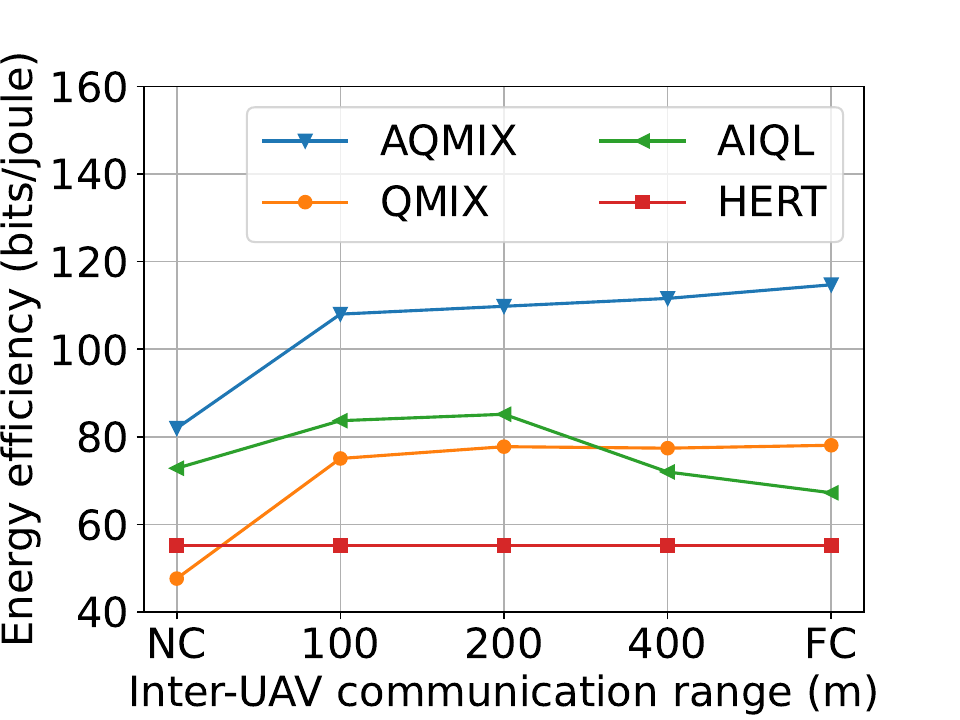}
         \caption{}
     \end{subfigure}
     \caption{Testing results of policies trained with different inter-UAV communication ranges. ``NC" - no communication between UAVs; ``FC" - full-communication to all UAVs.}
     \label{fig:com}
     \vspace{-0.4cm}
\end{figure}
Fig.~\ref{fig:com} plots the completion time and energy efficiency as functions of the inter-UAV communication range, which varies from no communication between UAVs to full communication, where UAVs can exchange information any time and from any position. Note that inter-UAV communication does not affect the performance of the heuristic method, as UAVs in this approach operate independently within preassigned collecting sub-areas. The results in Fig.~\ref{fig:com} show that extending the communication range enhances the performance of QMIX and AQMIX, demonstrated by the reduced completion time and the increased energy efficiency. This improvement occurs because a longer communication range allows UAVs to better synchronize their observations and assist each other in verifying mission completion, thereby reduces the total energy consumption. In contrast, increasing the communication range does not benefit AIQL and may even amplify non-stationarity in its learning process. These findings further confirm the effectiveness of the cooperation strategies learned by QMIX and AQMIX. Overall, AQMIX proves to be the superior approach.
\begin{figure}[h]
    \centering
    \includegraphics[width=0.6\linewidth]{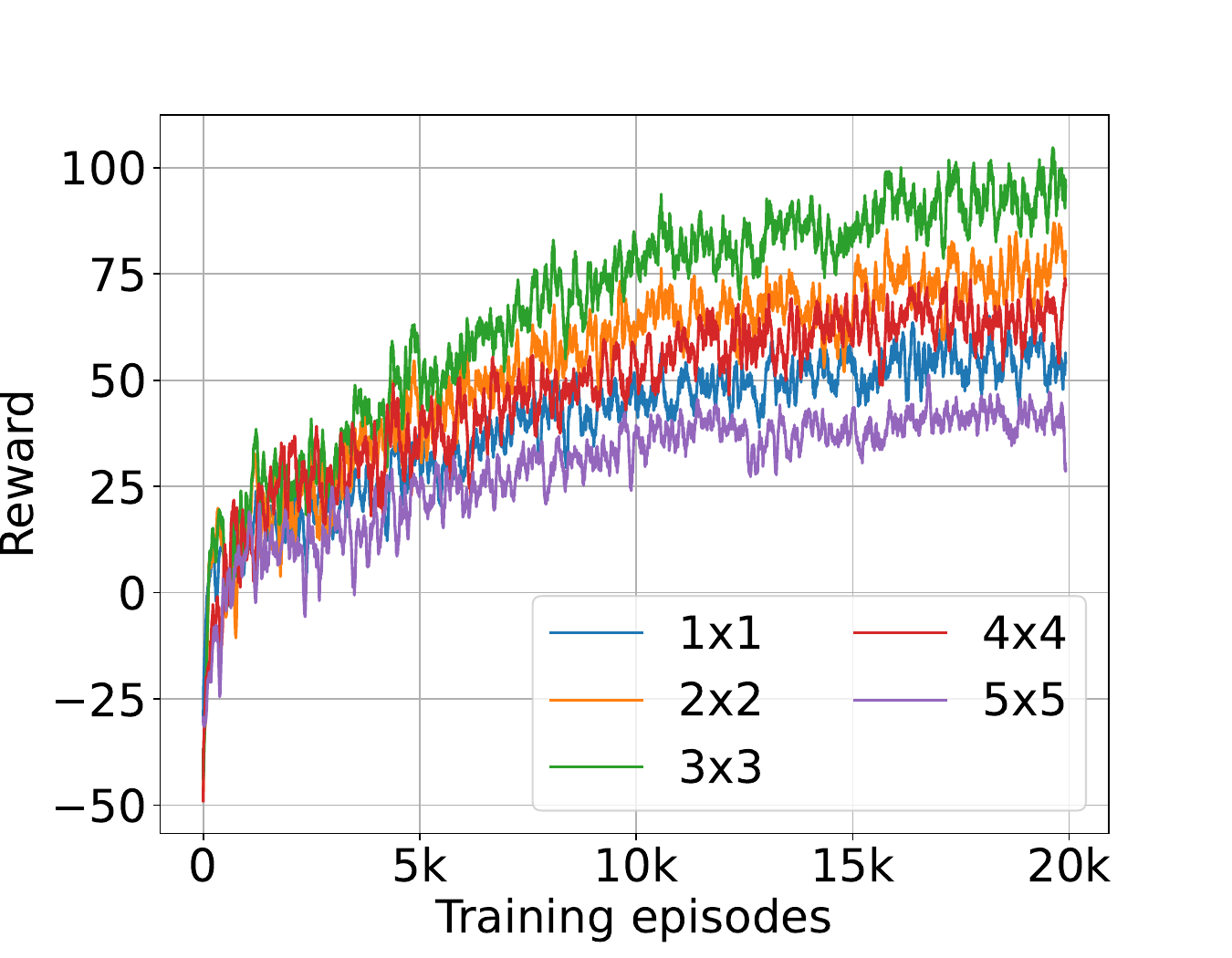}
    \caption{Learning curves of AQMIX with different kernel sizes used in state downsampling.}
    \label{fig:kernel_size}
\end{figure}

\subsubsection{Impacts of the state downsampling}
In Fig.~\ref{fig:kernel_size}, we report the results on a 15×15 grid with 2 UAVs, testing kernel sizes of 1×1 (no downsampling), 2×2, 3×3, 4×4, and 5×5. The results show that very small or very large kernels either overload the models with excessive detail or discard important spatial information, both leading to degraded policy performance. Based on these observations, a practical guideline for selecting the kernel size is to choose a value no larger than the agent’s observable region, ideally matching the size of this region. In our setup, each UAV observes a 3×3 window, and a 3×3 kernel provided the best trade-off between state compactness and information retention.

\subsubsection{Impacts of channel estimation error and bandwidth optimization} Fig.~\ref{fig:bandwidth_opt} plots the total time spent by all UAVs for collecting data and the energy efficiency of the best policy learned by AQMIX in the testing phase on an area of $8\times 8$ cells. The proposed optimal bandwidth allocation is compared with the equal bandwidth allocation counterpart, where each active SN is assigned an equal frequency bandwidth. The advantage of the proposed bandwidth optimization is clearly demonstrated via superior performance compared to the equal allocation scheme. Furthermore, the robustness of the proposed optimization is also confirmed via the operation under different CSI errors. 

\begin{figure}[t]
     \centering
     \begin{subfigure}{0.49\columnwidth}
         \centering
         \includegraphics[width=\textwidth]{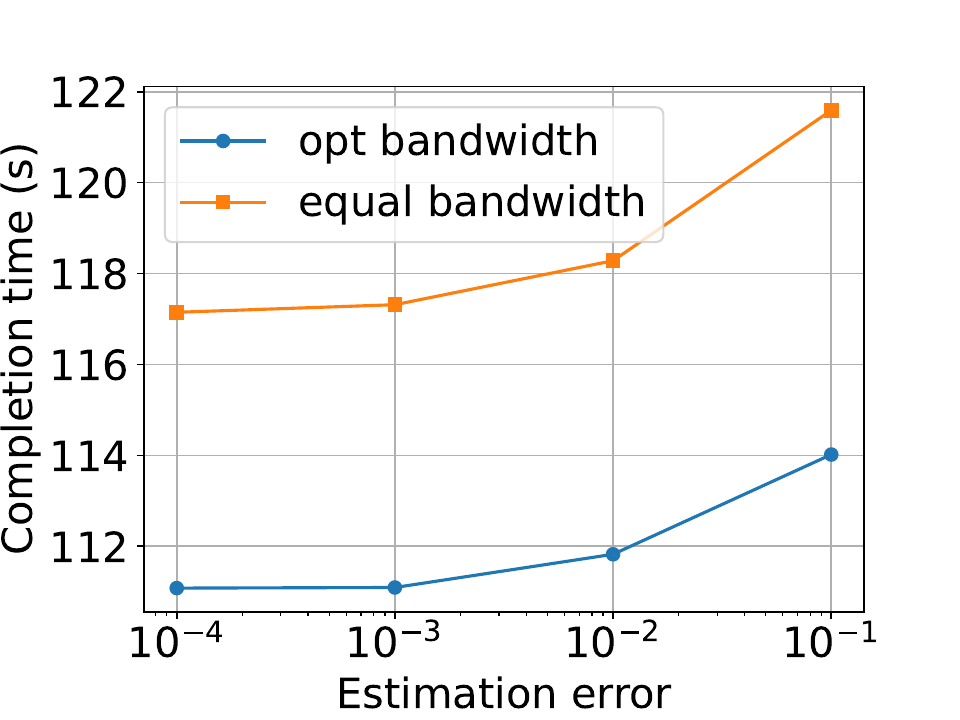}
         \caption{}
     \end{subfigure}
     \hfill
     \begin{subfigure}{0.49\columnwidth}
         \centering
         \includegraphics[width=\textwidth]{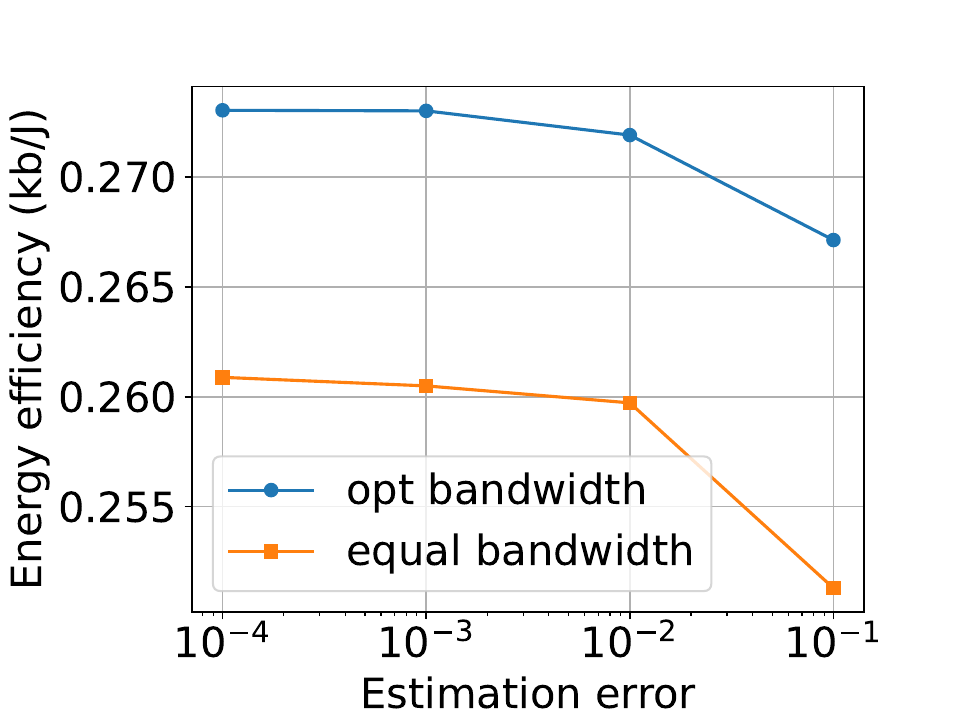}
         \caption{}
     \end{subfigure}
     \caption{Impact of channel estimation error to the performance of AQMIX.}
     \label{fig:bandwidth_opt}
     \vspace{-0.4cm}
\end{figure}

 \begin{figure}[t]
     \centering
     \begin{subfigure}{0.49\columnwidth}
         \centering
         \includegraphics[width=\textwidth]{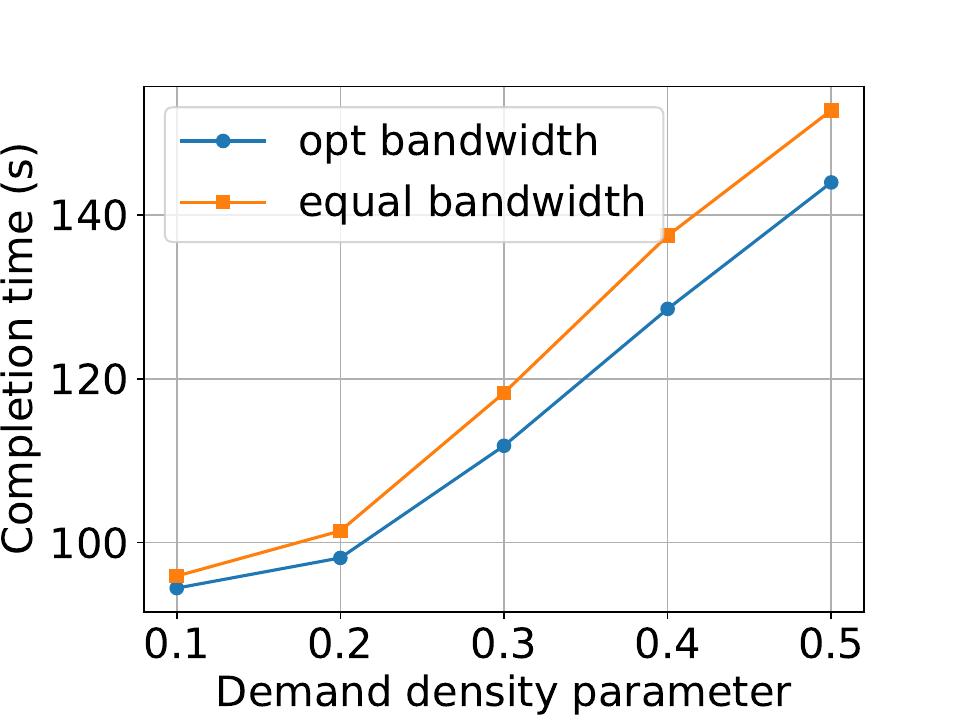}
         \caption{}
     \end{subfigure}
     \hfill
     \begin{subfigure}{0.49\columnwidth}
         \centering
         \includegraphics[width=\textwidth]{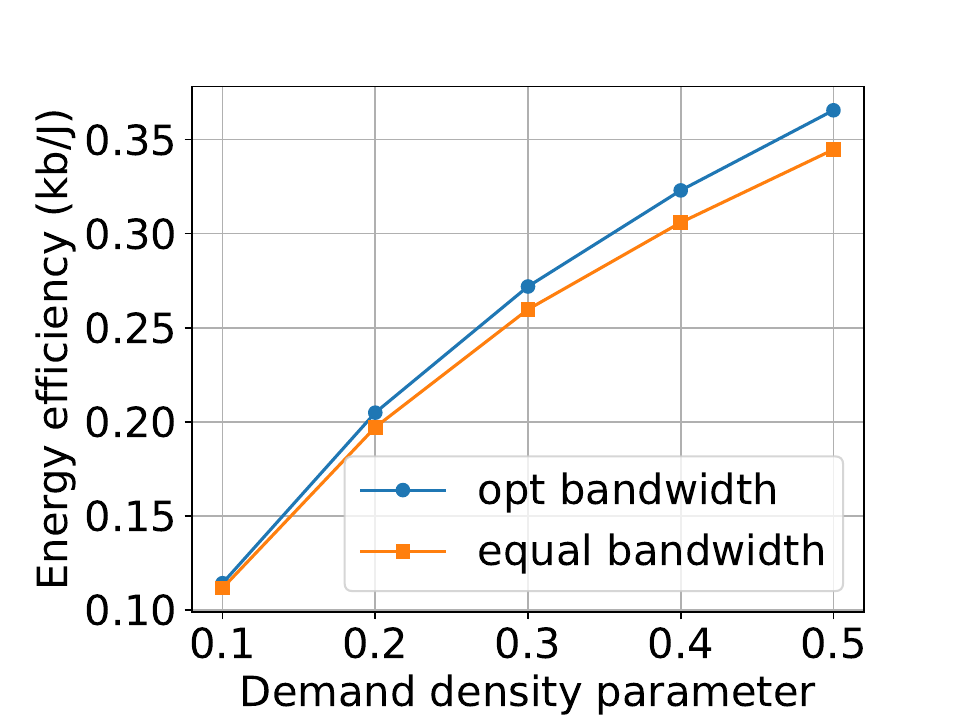}
         \caption{}
     \end{subfigure}
     \caption{Impact of bandwidth optimization to performance of AQMIX at different levels of collection demands, with estimation error of $\sigma_e^2 = 0.01$.}
     \label{fig:bandwidth_opt_density}
     \vspace{-0.4cm}
\end{figure}
Fig.~\ref{fig:bandwidth_opt_density}  compares the proposed bandwidth optimization with the equal allocation for various data density parameters. While the performance gain is marginal in low data densities, it becomes substantial for denser data scenarios. This is because, when data availability is spare, the collection time is relatively short compared to the total operating time, making the gain from reducing the communication time negligible. On the other hand, when there are more data to collect, the impact of bandwidth optimization becomes more pronounced, leading to noticeable performance improvements.

Finally, Fig.~\ref{fig:trajectories} visualizes the trajectories generated by all learned policies for one of the testing scenarios, on the collecting area of 8$\times$8 cells with two UAVs.
This figure once again highlights the superiority of the proposed method, AQMIX, as demonstrated by the cooperation between UAVs, which helps avoid trajectory overlaps between different UAVs, as observed in the other algorithms.

\begin{figure*}[h]
     \centering
     \begin{subfigure}{0.3\textwidth}
         \centering
         \includegraphics[width=\textwidth]{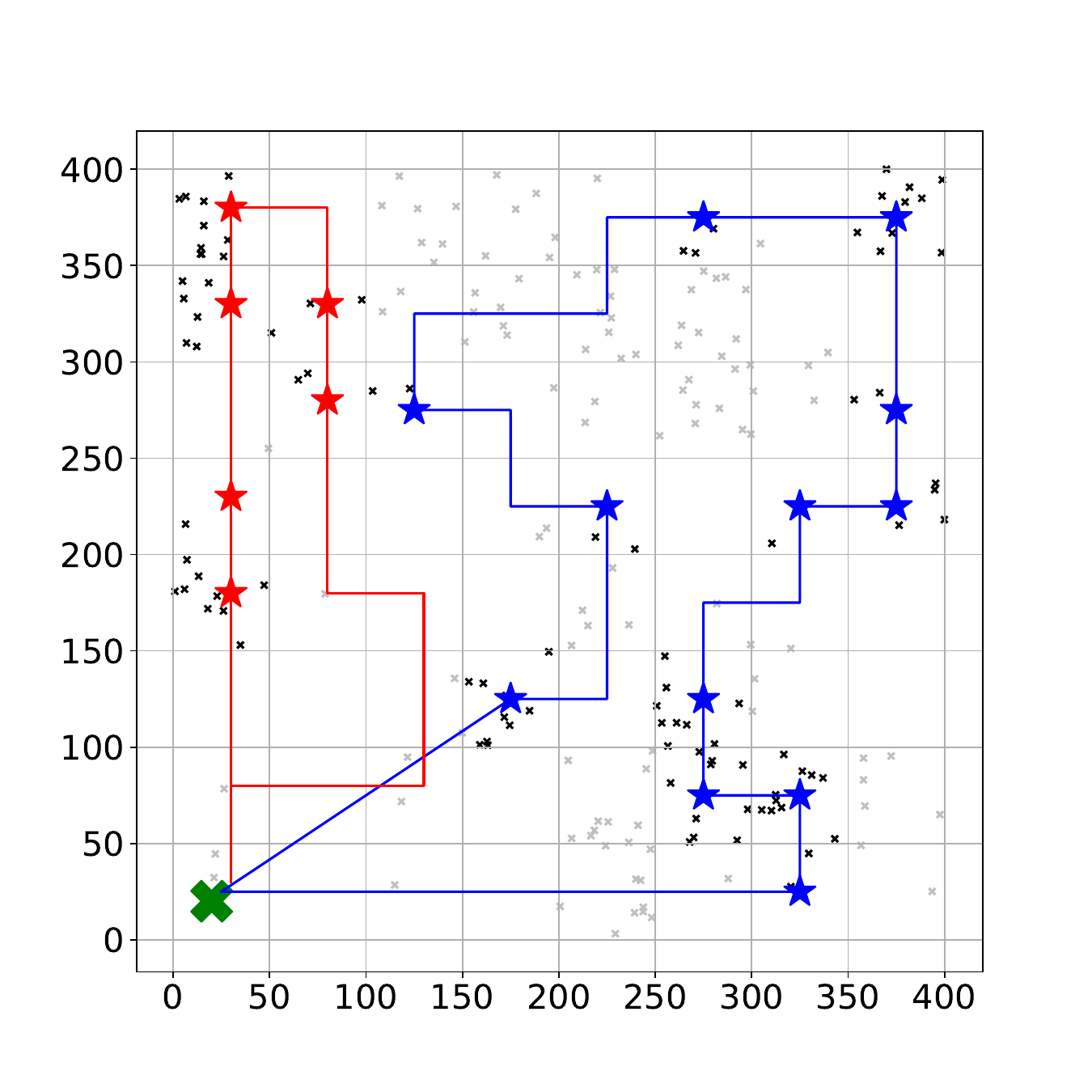}
         \caption{Proposed AQMIX}
     \end{subfigure}
     \hfill
     \begin{subfigure}{0.3\textwidth}
         \centering
         \includegraphics[width=\textwidth]{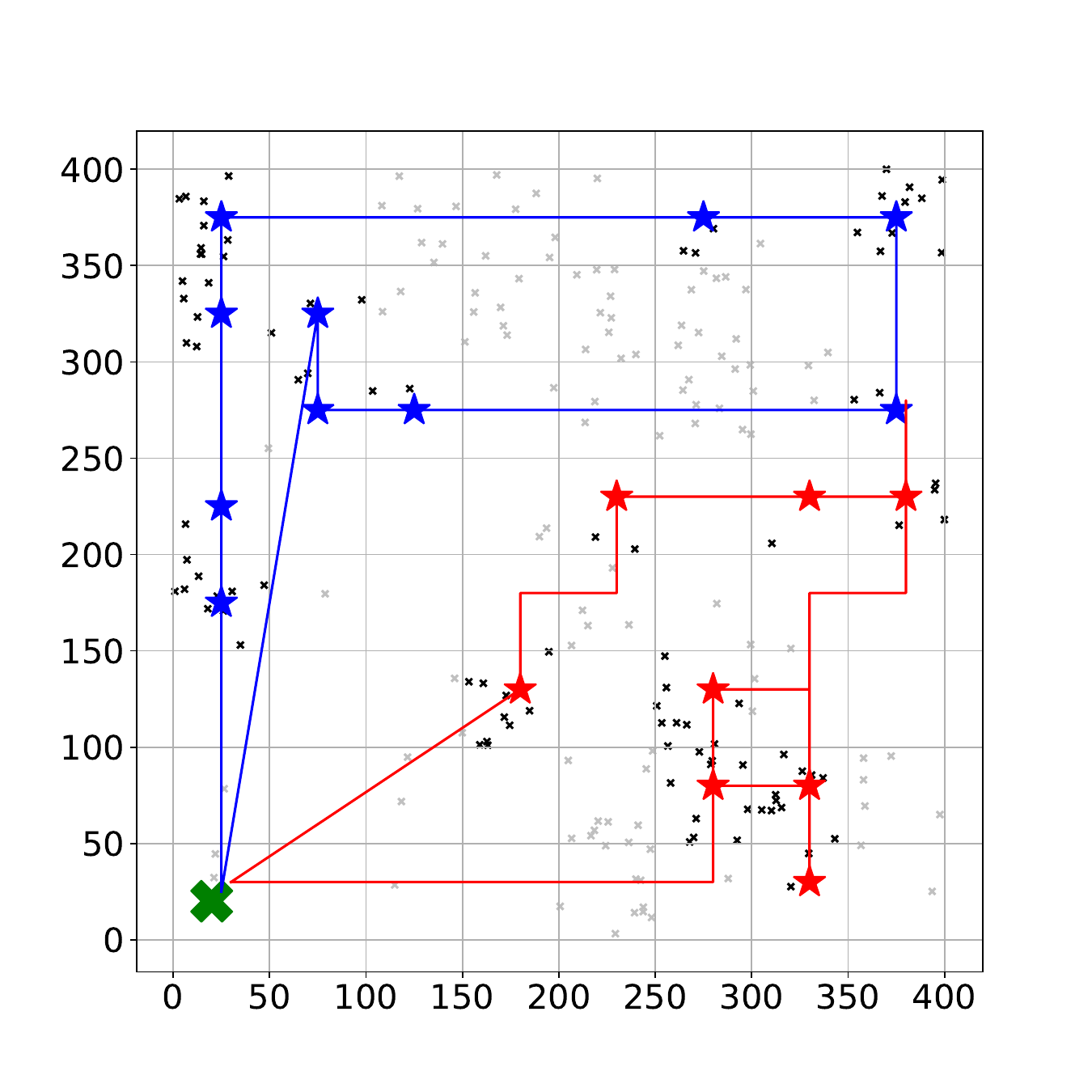}
         \caption{QMIX}
     \end{subfigure}
     \hfill
     \begin{subfigure}{0.3\textwidth}
         \centering
         \includegraphics[width=\textwidth]{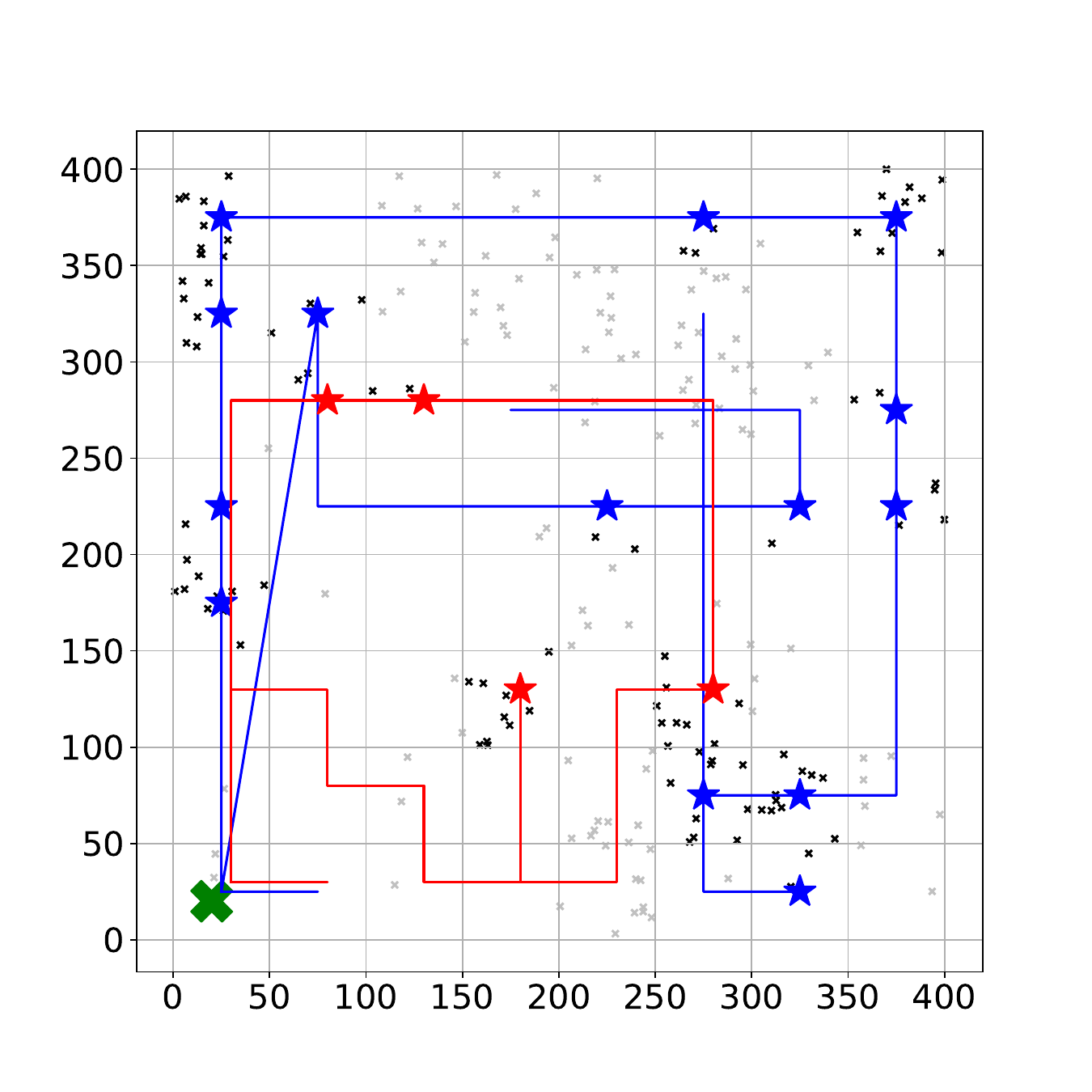}
         \caption{AIQL}
     \end{subfigure}
     \caption{Trajectories generated by learned policies. Green `X' symbols represent UAV's initial/final locations, stars represent hovering locations, small black and gray dots represents SNs containing and non-containing data, respectively.}
     \label{fig:trajectories}
     \vspace{-0.5cm}
\end{figure*}

\balance
\section{Conclusions} 
\label{sec:conclusion}
In this work, we have presented a solution to the problem of cooperative UAV data collection under realistic conditions characterized by asynchronization in the learning environment, primarily driven by stochastic data availability and limited inter-UAV communication. A key challenge addressed in this study is the incomplete information and asynchronous decision-making among UAVs, which are inherent to such scenarios. To tackle these issues, we introduced an asynchronous multi-agent learning framework, AQMIX, which has demonstrated superior performance and robustness compared to existing reference schemes. Moreover, we conducted a thorough sensitivity analysis of the proposed framework, evaluating its performance under varying system parameters such as communication range, energy budget, and data density. Our results highlight the adaptability and robustness of the framework across a wide range of operational conditions.

While the results presented in this paper are promising, there remains a gap between our solution and its practical implementation. One of the key challenges is the design of an efficient UAV-UAV communication protocol, particularly in determining when a UAV should initiate communication with another and how it should respond while simultaneously moving or collecting data.
Another topic is to further enhance the bandwidth utilization efficiency, by allowing the UAVs to operate the full system bandwidth. In this case, inter-UAV interference can be approximated from the channel statistics when computing the UAV hovering time.
Beyond the specific case study examined here, we believe that our novel modeling approach and proposed model-free learning solution can offer valuable insights and be applied to other systems where asynchrony exists.


%
%
%


\begin{thebibliography}{10}

\bibitem{cuong24}
C.~Le, T.~X. Vu, and S.~Chatzinotas, ``Cooperative UAVs with asynchronous learning for remote data collection,'' in {\em IEEE Global Communications Conference}, 2024.

\bibitem{bayerlein2021multi}
H.~Bayerlein, M.~Theile, M.~Caccamo, and D.~Gesbert, ``Multi-UAV path planning for wireless data harvesting with deep reinforcement learning,'' {\em IEEE Open J. Commun. Soc.}, vol.~2, pp.~1171--1187, 2021.

\bibitem{oubbati2022synchronizing}
O.~S. Oubbati, M.~Atiquzzaman, H.~Lim, A.~Rachedi, and A.~Lakas, ``Synchronizing UAV teams for timely data collection and energy transfer by deep reinforcement learning,'' {\em IEEE Trans. Veh. Techno.}, vol.~71, no.~6, pp.~6682--6697, 2022.

\bibitem{li2022learning}
Z.~Li, P.~Tong, J.~Liu, X.~Wang, L.~Xie, and H.~Dai, ``Learning-based data gathering for information freshness in UAV-assisted IoT networks,'' {\em IEEE Internet Things J.}, vol.~10, no.~3, pp.~2557--2573, 2022.

\bibitem{chen2022joint}
G.~Chen, X.~B. Zhai, and C.~Li, ``Joint optimization of trajectory and user association via reinforcement learning for UAV-aided data collection in wireless networks,'' {\em IEEE Trans. Wirel. Commun.}, 2022.

\bibitem{wang2023cooperative}
X.~Wang, M.~Yi, J.~Liu, Y.~Zhang, M.~Wang, and B.~Bai, ``Cooperative data collection with multiple UAVs for information freshness in the internet of things,'' {\em IEEE Trans. Commun.}, 2023.

\bibitem{rashid2020monotonic}
T.~Rashid, M.~Samvelyan, C.~S. De~Witt, G.~Farquhar, J.~Foerster, and S.~Whiteson, ``Monotonic value function factorisation for deep multi-agent reinforcement learning,'' {\em J. Mach. Learn. Res.}, vol.~21, no.~178, pp.~1--51, 2020.

\bibitem{yuan2022joint}
X.~Yuan, Y.~Hu, J.~Zhang, and A.~Schmeink, ``Joint user scheduling and UAV trajectory design on completion time minimization for UAV-aided data collection,'' {\em IEEE Trans. Wirel. Commun.}, 2022.

\bibitem{zhan2017energy}
C.~Zhan, Y.~Zeng, and R.~Zhang, ``Energy-efficient data collection in UAV enabled wireless sensor network,'' {\em IEEE Wirel. Commun. Lett.}, vol.~7, no.~3, pp.~328--331, 2017.

\bibitem{wang2019energy}
Z.~Wang, R.~Liu, Q.~Liu, J.~S. Thompson, and M.~Kadoch, ``Energy-efficient data collection and device positioning in UAV-assisted IoT,'' {\em IEEE Internet Things J.}, vol.~7, no.~2, pp.~1122--1139, 2019.

\bibitem{you20193d}
C.~You and R.~Zhang, ``3d trajectory optimization in rician fading for UAV-enabled data harvesting,'' {\em IEEE Trans. Wirel. Commun.}, vol.~18, no.~6, pp.~3192--3207, 2019.

\bibitem{sun2023max}
C.~Sun, X.~Xiong, Z.~Zhai, W.~Ni, T.~Ohtsuki, and X.~Wang, ``Max-min fair 3d trajectory design and transmission scheduling for solar-powered fixed-wing UAV-assisted data collection,'' {\em IEEE Trans. Wirel. Commun.}, 2023.

\bibitem{samir2019uav}
M.~Samir, S.~Sharafeddine, C.~M. Assi, T.~M. Nguyen, and A.~Ghrayeb, ``UAV trajectory planning for data collection from time-constrained IoT devices,'' {\em IEEE Trans. Wirel. Commun.}, vol.~19, no.~1, pp.~34--46, 2019.

\bibitem{tran2021uav}
D.-H. Tran, V.-D. Nguyen, S.~Chatzinotas, T.~X. Vu, and B.~Ottersten, ``UAV relay-assisted emergency communications in IoT networks: Resource allocation and trajectory optimization,'' {\em IEEE Trans. Wirel. Commun.}, vol.~21, no.~3, pp.~1621--1637, 2021.

\bibitem{feng2021uav}
T.~Feng, L.~Xie, J.~Yao, and J.~Xu, ``UAV-enabled data collection for wireless sensor networks with distributed beamforming,'' {\em IEEE Trans. Wirel. Commun.}, vol.~21, no.~2, pp.~1347--1361, 2021.

\bibitem{du2023time}
P.~Du, F.~Xie, S.~Chen, and X.~Zhang, ``Time-constrained UAV-aided data collection for IoT networks with energy harvesting,'' in {\em IEEE Conference on Computer Communications Workshops}, pp.~1--6, IEEE, 2023.

\bibitem{hu2020aoi}
H.~Hu, K.~Xiong, G.~Qu, Q.~Ni, P.~Fan, and K.~B. Letaief, ``AoI-minimal trajectory planning and data collection in UAV-assisted wireless powered IoT networks,'' {\em IEEE Internet Things J.}, vol.~8, no.~2, pp.~1211--1223, 2020.

\bibitem{zhan2019completion}
C.~Zhan and Y.~Zeng, ``Completion time minimization for multi-UAV-enabled data collection,'' {\em IEEE Trans. Wirel. Commun.}, vol.~18, no.~10, pp.~4859--4872, 2019.

\bibitem{zhan2019aerial}
C.~Zhan and Y.~Zeng, ``Aerial--ground cost tradeoff for multi-UAV-enabled data collection in wireless sensor networks,'' {\em IEEE Trans. Commun.}, vol.~68, no.~3, pp.~1937--1950, 2019.

\bibitem{zhang2021minimizing}
J.~Zhang, Z.~Li, W.~Xu, J.~Peng, W.~Liang, Z.~Xu, X.~Ren, and X.~Jia, ``Minimizing the number of deployed UAVs for delay-bounded data collection of IoT devices,'' in {\em IEEE Conference on Computer Communications}, pp.~1--10, IEEE, 2021.

\bibitem{xu2021minimizing}
W.~Xu, T.~Xiao, J.~Zhang, W.~Liang, Z.~Xu, X.~Liu, X.~Jia, and S.~K. Das, ``Minimizing the deployment cost of UAVs for delay-sensitive data collection in IoT networks,'' {\em IEEE/ACM Trans. Netw. }, vol.~30, no.~2, pp.~812--825, 2021.

\bibitem{arsham25TOCD}
A.~Mostaani, T.~X.~Vu, H.~Habibi, S.~Chatzinotas and B.~Ottersten, ``Task-Oriented Communication Design at Scale,'' {\em IEEE Trans. Commun.}, vol.~73, no.~1, pp.~378--393, Jan.~2025.

\bibitem{ding20203d}
R.~Ding, F.~Gao, and X.~S. Shen, ``3D UAV trajectory design and frequency band allocation for energy-efficient and fair communication: A deep reinforcement learning approach,'' {\em IEEE Trans. Wirel. Commun.}, vol.~19, no.~12, pp.~7796--7809, 2020.

\bibitem{wang2021trajectory}
Y.~Wang, Z.~Gao, J.~Zhang, X.~Cao, D.~Zheng, Y.~Gao, D.~W.~K. Ng, and M.~Di~Renzo, ``Trajectory design for UAV-based internet of things data collection: A deep reinforcement learning approach,'' {\em IEEE Internet Things J.}, vol.~9, no.~5, pp.~3899--3912, 2021.

\bibitem{fan2022ris}
X.~Fan, M.~Liu, Y.~Chen, S.~Sun, Z.~Li, and X.~Guo, ``Ris-assisted UAV for fresh data collection in 3D urban environments: A deep reinforcement learning approach,'' {\em IEEE Trans. Veh. Techno.}, vol.~72, no.~1, pp.~632--647, 2022.

\bibitem{hu2020cooperative}
J.~Hu, H.~Zhang, L.~Song, R.~Schober, and H.~V. Poor, ``Cooperative internet of UAVs: Distributed trajectory design by multi-agent deep reinforcement learning,'' {\em IEEE Trans. Commun.}, vol.~68, no.~11, pp.~6807--6821, 2020.

\bibitem{oubbati2021multi}
O.~S. Oubbati, M.~Atiquzzaman, A.~Lakas, A.~Baz, H.~Alhakami, and W.~Alhakami, ``Multi-UAV-enabled aoi-aware wpcn: A multi-agent reinforcement learning strategy,'' in {\em IEEE Conference on Computer Communications Workshops}, pp.~1--6, 2021.

\bibitem{emami2023aoi}
Y.~Emami, K.~Li, Y.~Niu, and E.~Tovar, ``AoI minimization using multi-agent proximal policy optimization in UAVs-assisted sensor networks,'' in {\em IEEE International Conference on Communications}, pp.~228--233, 2023.

\bibitem{messaoudi2023uav}
K.~Messaoudi, O.~S. Oubbati, A.~Rachedi, and T.~Bendouma, ``UAV-UGV-based system for aoi minimization in IoT networks,'' in {\em IEEE International Conference on Communications}, pp.~4743--4748,  2023.

\bibitem{bernstein2002complexity}
D.~S. Bernstein, R.~Givan, N.~Immerman, and S.~Zilberstein, ``The complexity of decentralized control of Markov decision processes,'' {\em Mathematics of operations research}, vol.~27, no.~4, pp.~819--840, 2002.

\bibitem{omidshafiei2015decentralized}
S.~Omidshafiei, A.-A. Agha-Mohammadi, C.~Amato, and J.~P. How, ``Decentralized control of partially observable Markov decision processes using belief space macro-actions,'' in {\em 2015 IEEE international conference on robotics and automation (ICRA)}, pp.~5962--5969, IEEE, 2015.

\bibitem{wang2021reducing}
J.~Wang and L.~Sun, ``Reducing bus bunching with asynchronous multi-agent reinforcement learning,'' {\em Proc. 3rd International Joint Conference on Artificial Intelligence}, pp.~426--433, 2021.

\bibitem{yu2023asynchronous}
C.~Yu, X.~Yang, J.~Gao, J.~Chen, Y.~Li, J.~Liu, Y.~Xiang, R.~Huang, H.~Yang, Y.~Wu, {\em et~al.}, ``Asynchronous multi-agent reinforcement learning for efficient real-time multi-robot cooperative exploration,'' in {\em Proc. International Conference on Autonomous Agents and Multiagent Systems}, pp.~1107--1115, 2023.

\bibitem{zeng2019energy}
Y.~Zeng, J.~Xu, and R.~Zhang, ``Energy minimization for wireless communication with rotary-wing UAV,'' {\em IEEE Trans. Wirel. Commun.}, vol.~18, no.~4, pp.~2329--2345, 2019.

\bibitem{al2014optimal}
A.~Al-Hourani, S.~Kandeepan, and S.~Lardner, ``Optimal lap altitude for maximum coverage,'' {\em IEEE Wirel. Commun. Lett.}, vol.~3, no.~6, pp.~569--572, 2014.

\bibitem{hausknecht2015deep}
M.~Hausknecht and P.~Stone, ``Deep recurrent q-learning for partially observable mdps,'' in {\em AAAI fall symposium series}, 2015.

\bibitem{cho2014learning}
K.~Cho, B.~Van~Merri{\"e}nboer, C.~Gulcehre, D.~Bahdanau, F.~Bougares, H.~Schwenk, and Y.~Bengio, ``Learning phrase representations using rnn encoder-decoder for statistical machine translation,'' {\em arXiv preprint arXiv:1406.1078}, 2014.

\end{thebibliography}

\end{document}